\documentclass[12pt]{article}

\usepackage[utf8]{inputenc}
\usepackage[T1]{fontenc}
\usepackage{lmodern}
\usepackage[margin=1in]{geometry}
\usepackage{setspace}
\usepackage{amsmath,amssymb,amsthm}
\usepackage{booktabs}
\usepackage{graphicx}
\usepackage{hyperref}
\usepackage{natbib}
\usepackage{footmisc}
\usepackage{enumitem}
\usepackage{amsthm,amsmath,amssymb}
\newtheorem{proposition}{Proposition}
\newtheorem{lemma}{Lemma}
\newtheorem{definition}{Definition}

\hypersetup{
    colorlinks=true,
    linkcolor=blue,
    citecolor=blue,
    urlcolor=blue
}

\title{\textbf{Statehood Without Capacity}}
\author{Rok Spruk\thanks{Associate Professor, School of Economics and Business, University of Ljubljana, Kardeljeva ploščad 17, SI-1000 Ljubljana, Slovenia. Email: \texttt{rok.spruk@ef.uni-lj.si}.}}
\date{April 7, 2026}

\begin{document}

\maketitle

\begin{abstract}
This paper develops a political-economy theory of \emph{statehood without capacity}. I argue that under specific institutional and geopolitical conditions, a polity can become trapped in an equilibrium of \emph{nominal statehood}: a state in which claims to sovereignty, external recognition, and symbolic legitimacy persist or even strengthen while the coercive, fiscal, administrative, and legal capacities required for effective statehood remain weak. The mechanism is driven by three forces. First, fragmented elites may privately benefit from preserving autonomous control, patronage, and localized rent extraction rather than consolidating authority into a unified state. Second, externally mediated transfers can reduce the immediate costs of institutional non-consolidation and thereby stabilize a low-capacity equilibrium. Third, international recognition and symbolic endorsement may be only weakly conditioned on domestic administrative performance, allowing recognition capital to accumulate more rapidly than capacity capital. The theory generates a dynamic divergence between juridical or symbolic statehood and effective statehood, with implications for investment, fiscal fragility, corruption, and vulnerability to conflict shocks. The paper derives testable predictions and then interprets Palestine as a flagship application of the broader mechanism. The central implication is that statehood is not only a question of recognition or territorial claim; it is an equilibrium outcome of institutional consolidation. Where the incentives to consolidate remain weak, sovereignty may be asserted without becoming viable.
\end{abstract}

\bigskip
\noindent \textbf{Keywords:} state capacity, nominal statehood, political fragmentation, sovereignty, institutional fragility, political economy, Palestine


\section{Introduction}

Modern international politics often treats statehood primarily as a question of recognition, sovereignty, and territorial claim. Political discourse is therefore naturally drawn to questions of rights, diplomatic legitimacy, and legal entitlement. Yet the comparative political economy of development points to a deeper and more demanding standard. A viable state is not merely a recognized political unit. It is an institutional achievement: a durable concentration of coercive authority, a fiscal apparatus capable of extraction and provision, an administration able to implement policy beyond the capital city, and a legal order sufficiently credible to sustain investment, contracting, and long-horizon expectations. Where these capacities remain thin, statehood may exist in formal or symbolic terms while remaining substantively fragile in economic and institutional terms.

This paper develops a theory of \emph{statehood without capacity}. Its central claim is that sovereignty and state capacity, though often treated as complements, can diverge persistently. Under specific political and economic conditions, a polity may become trapped in an equilibrium of \emph{nominal statehood}, an equilibrium in which claims to statehood, external recognition, and symbolic legitimacy remain politically salient while the coercive, fiscal, administrative, and legal capacities required for effective statehood fail to consolidate. In this equilibrium, statehood is advanced as a juridical or moral object, but the institutional foundations of self-sustaining governance remain weak.

The core mechanism is simple. State-building requires consolidation. Rival elites must give up autonomous command, accept tighter constraints on private extraction, and internalize a larger share of the gains from common administration than from fragmented control. But these conditions often fail. Where factional leaders derive substantial private returns from local coercive autonomy, patronage networks, discretionary access to public resources, or externally mediated transfers, institutional fragmentation may be socially inefficient yet privately optimal. If external actors further soften the cost of non-consolidation through humanitarian assistance, fiscal support, or diplomatic engagement that is only weakly conditioned on genuine administrative integration, the fragmented equilibrium may persist. The result is a polity in which recognition advances faster than capacity.

The paper contributes to several literatures at once. First, it speaks to the political economy of state capacity. A large body of work shows that the emergence of capable states depends on the joint development of fiscal extraction, administrative reach, and monopoly over violence, and that these capacities are central to long-run development \citep{BesleyPersson2009, BesleyPersson2010, DinceccoPrado2012, JohnsonKoyama2017}. This stream of literature has shown convincingly that states do not become effective merely because they are recognized. They also become effective when rulers acquire both the incentive and the organizational means to centralize authority and invest in public capacity. Yet much of this work implicitly assumes that sovereignty and capacity are either jointly produced or at least mutually reinforcing. This paper identifies a class of cases in which the two may come apart.

Second, the paper contributes to the literature on weak states, fragmented sovereignty, and political order. Scholars of comparative politics have long noted that many polities possess the external trappings of statehood while lacking the internal coherence normally associated with state formation \citep{JacksonRosberg1982, Migdal1988, Herbst2000}. The present paper sharpens that insight by giving it a formal political-economy structure. Rather than treating weakness as a residual condition, I model it as an equilibrium outcome sustained by the incentives of domestic elites and the stabilizing role of external actors. In that sense, the paper moves from descriptive weak-state analysis to an explicit theory of why institutional non-consolidation can persist even when all parties understand its aggregate costs.

Third, the paper connects to literature on rentier politics, aid dependence, and the fiscal foundations of accountability. One of the classic insights of political economy is that taxation and representation are linked. Rulers who depend on domestic extraction must bargain more intensively with society and face stronger incentives to build administrative competence than rulers whose revenues derive from external rents \citep{Mahdavy1970, Ross2001, Moore2004}. Externally financed political orders may therefore display a distinctive pathology. Outside resources prevent immediate collapse while also weakening the pressure for fiscal-contract formation and institutional reform. This paper extends that logic to the domain of contested or incomplete statehood. External transfers need not be malign to have this effect. It is enough that they stabilize continuity under fragmentation more effectively than they reward costly institutional integration.

Fourth, the paper relates to the economics of conflict and institutional fragility. Civil conflict and political violence do not merely destroy output in the short run; they also alter incentives to invest, shorten planning horizons, distort public finance, and reallocate political value toward coercive control and survival \citep{FearonLaitin2003, BlattmanMiguel2010, BesleyPersson2011}. Where authority is divided and credible commitment among elites is weak, conflict risk amplifies the attractiveness of fragmentation relative to centralization. State-building becomes not only administratively costly but strategically dangerous. In such settings, even the prospect of future efficiency gains may fail to induce reform if elites cannot credibly secure their position within the post-consolidation order.

The paper’s central theoretical innovation is to distinguish between two forms of political capital: \emph{recognition capital} and \emph{capacity capital}. Recognition capital denotes the stock of external legitimacy, diplomatic endorsement, juridical acceptance, and symbolic support attached to a political project of statehood. Capacity capital denotes the stock of coercive, fiscal, administrative, and legal capabilities that make statehood function in practice. In many successful cases of state formation, these two forms of capital move together. Recognition follows capacity, or capacity follows recognition under strong complementarity. But this complementarity is not guaranteed. The argument developed here is that under some conditions recognition capital can accumulate even while capacity capital stagnates. When that happens, a polity may become increasingly visible and validated as a state project without becoming increasingly capable as a governing order.

This divergence has first-order economic consequences. A polity trapped in nominal statehood should display low investment, high volatility, weak domestic revenue capacity, persistent corruption, shallow productive transformation, and heightened vulnerability to shocks. The reason is straightforward. Firms and households respond not to flags or declarations but to enforceable rules, credible contracts, predictable taxation, and the expectation that current bargains will survive the next political rupture. Where those conditions remain weak, economic life becomes short-horizon, defensive, and dependent on externally mediated stabilization. The resulting equilibrium is not merely underdevelopment in a generic sense. It is underdevelopment linked to the political persistence of non-consolidation.

The paper then interprets Palestine as a flagship application of the theory. The Palestinian case is not analytically important because it is unique. It is important because it brings the mechanism into unusually sharp relief. Palestinian politics has long combined divided authority, weak performance-based accountability, externally mediated finance, limited unified control over coercion, and high levels of symbolic and diplomatic salience. These features are often treated as transitory imperfections around an otherwise linear state-building process. The argument here is different. They are better understood as the components of a stable low-capacity equilibrium in which fragmented authority remains privately sustainable even though it is socially costly. Palestine thus serves as an analytically revealing application of a more general theory rather than as the theory’s sole object.

The broader implication of the analysis is simple but consequential. Statehood is not only declared, recognized, or morally endorsed. It is built. When the incentives for institutional consolidation are weak, when the private returns to fragmentation remain high, and when external actors cushion the costs of non-consolidation, sovereignty may become politically durable without becoming institutionally viable. Understanding that divergence is essential for any serious political economy of fragile statehood.

The paper proceeds as follows. Section 2 defines statehood, capacity, and nominal statehood, and introduces the distinction between recognition capital and capacity capital. Section 3 presents a dynamic political-economy model in which rival elites choose between institutional consolidation and fragmented control under external transfers and imperfectly capacity-conditioned recognition. Section 4 characterizes equilibrium and derives comparative statics. Section 5 develops the theory’s testable implications. Section 6 interprets Palestine as a flagship application of the broader mechanism. Section 7 discusses wider implications for the political economy of development, fragile states, and state-building under incomplete sovereignty. Section 8 concludes.

\section{Statehood, Capacity, and Nominal Statehood}

\subsection{Statehood as an institutional achievement}

Political and legal discourse often treats statehood as a question of recognition, sovereignty, and territorial claim. In the political economy of development, however, statehood is more usefully understood as an institutional achievement. A viable state is not merely a claimant to authority. It is an organization capable of exercising authority in a durable, rule-bound, and developmentally consequential manner. At a minimum, this requires four mutually reinforcing capacities: the ability to centralize coercion, the ability to extract and allocate fiscal resources, the ability to administer policies across territory, and the ability to sustain legal credibility in ways that lengthen private planning horizons. These dimensions are analytically distinct, but they tend to be complementary in practice. Fiscal extraction is difficult without coercive control, administrative reach is difficult without revenue, legal credibility is difficult without a minimally autonomous enforcement apparatus, and all of them interact to determine whether private actors expect current bargains to survive future shocks.

Let effective statehood be represented by
\[
S_t = \mathcal{S}(C_t,F_t,A_t,L_t),
\]
where \(C_t\) denotes coercive centralization, \(F_t\) fiscal capacity, \(A_t\) administrative capability, and \(L_t\) legal credibility at time \(t\). Assume that \(\mathcal{S}(\cdot)\) is increasing in each argument and weakly supermodular:
\[
\frac{\partial \mathcal{S}}{\partial x_j} > 0
\qquad \text{for all } x_j \in \{C_t,F_t,A_t,L_t\},
\]
and
\[
\frac{\partial^2 \mathcal{S}}{\partial x_j \partial x_k} \geq 0
\qquad \text{for all } j \neq k.
\]
The supermodularity assumption captures the idea that partial improvements in one margin of capacity are more valuable when complementary margins are also strong. A polity with modest coercive control but no fiscal apparatus, or with formal courts but no enforcement capacity, does not generate the full benefits normally associated with statehood.

This way of defining statehood has two advantages. First, it clarifies that statehood is not exhausted by its legal shell. Second, it links the concept of statehood directly to the economics of growth, investment, and institutional persistence. The relevant question is not only whether a polity is recognized as a state or as a future state, but whether it has crossed a threshold of institutional coherence sufficient to sustain a self-enforcing order.

\subsection{Recognition capital and capacity capital}

The central conceptual distinction of the analysis is between \emph{recognition capital} and \emph{capacity capital}. These are related but non-identical stocks.

Let \(R_t \geq 0\) denote recognition capital. This is captured by the accumulated stock of diplomatic legitimacy, juridical acceptance, symbolic endorsement, and political validation attached to a project of statehood. Recognition capital may be increased through formal recognition, participation in international organizations, diplomatic support, narrative legitimacy, or broader forms of external endorsement.

Let \(K_t \geq 0\) denote capacity capital. This is the accumulated stock of coercive, fiscal, administrative, and legal capability that supports effective governance. It is the institutional asset base that determines whether statehood is viable in a functional rather than purely symbolic sense.

Formally, let
\[
K_t = \mathcal{K}(C_t,F_t,A_t,L_t),
\]
where \(\mathcal{K}(\cdot)\) is increasing in each argument. For simplicity, one may think of \(K_t\) as an index of institutional capability, while \(R_t\) indexes external and symbolic validation.

In many successful cases of state formation, \(R_t\) and \(K_t\) move together. Diplomatic recognition follows institutional consolidation, or else early recognition accelerates capacity accumulation because it unlocks credible access to trade, finance, and administrative cooperation. In either case, the two stocks are complementary. But complementarity is not guaranteed. The key possibility explored here is that \(R_t\) may increase while \(K_t\) stagnates or declines:
\[
\frac{dR_t}{dt} > 0
\qquad \text{and} \qquad
\frac{dK_t}{dt} \leq 0.
\]
This divergence is the core pathology of statehood without capacity.

To sharpen the distinction, define the \emph{recognition-capacity gap} as
\[
G_t \equiv R_t - \psi K_t,
\]
for some scaling parameter \(\psi > 0\). A rise in \(G_t\) indicates that recognition is outpacing institutional capability. The argument of the paper is that in some political environments \(G_t\) is not merely a transition artifact. It can be an equilibrium outcome.

\subsection{Nominal statehood}

The concept of \emph{nominal statehood} captures the case in which a polity possesses substantial recognition capital but insufficient capacity capital to sustain effective statehood.

\begin{definition}
A polity is in a regime of nominal statehood at time \(t\) if
\[
R_t \geq \bar{R}
\qquad \text{and} \qquad
K_t < \bar{K},
\]
where \(\bar{R}\) is a threshold level of recognition sufficient for statehood claims to be politically and diplomatically salient, and \(\bar{K}\) is a threshold level of institutional capacity sufficient for viable statehood.
\end{definition}

This definition is intentionally minimalist. It does not require full formal recognition, nor does it require complete institutional collapse. It identifies an intermediate but politically important condition. Namely, the polity is sufficiently recognized to function as a meaningful state project, but insufficiently capable to function as a self-sustaining state.

Nominal statehood should not be confused with mere weakness. Many weak states remain states in a robust sense because they retain a unified monopoly of violence, a functioning fiscal apparatus, and administrative continuity despite low income or corruption. Nor should nominal statehood be confused with pure statelessness. The focal object of interest here is precisely the unstable middle as a politically salient statehood project that has acquired enough recognition to endure symbolically, yet not enough capacity to become institutionally self-enforcing.

\subsection{Why recognition and capacity may diverge?}

Why might \(R_t\) and \(K_t\) diverge? The simplest answer is that they are produced by different mechanisms and rewarded by different constituencies. Recognition capital is often generated externally. Diplomatic actors, international organizations, transnational advocacy networks, and foreign governments may reward claims to statehood for reasons linked to legal principle, normative commitment, conflict management, historical responsibility, strategic balancing, or symbolic politics. Those actors may care about capacity, but not necessarily at the margin that determines recognition.

Capacity capital, by contrast, is generated internally through coercive consolidation, fiscal bargaining, bureaucratic investment, and legal institutionalization. These are costly processes. They require elites to surrender discretion, accept more transparent constraints, reduce private rents, and invest in administrative structures whose returns are diffuse and long-run. Capacity therefore depends on a domestic incentive structure that may be much harder to alter than the external environment governing recognition.

This asymmetry matters because it implies that the political return to recognition may remain high even when the political return to capacity-building is weak. A polity can therefore accumulate recognition capital without resolving the internal collective-action problem of state formation.

\subsection{A simple benchmark}

To clarify the intuition, suppose that state viability is given by
\[
V_t = \Lambda(R_t,K_t),
\]
where \(\Lambda_R > 0\), \(\Lambda_K > 0\), and
\[
\Lambda_{RK} > 0.
\]
Recognition and capacity are thus complementary in producing viability. Yet private elite payoffs need not weight them symmetrically. Let the representative elite valuation of political order be
\[
U_t^{e} = \mu R_t + \nu K_t + \Xi_t,
\]
where \(\mu,\nu > 0\), but where \(\Xi_t\) captures faction-specific control rents, autonomous authority, and other private benefits that may be increasing in fragmentation. If \(\Xi_t\) is sufficiently large under non-consolidation, elites may prefer an increase in \(R_t\) without a corresponding increase in \(K_t\). Recognition then substitutes politically for capacity even when it does not substitute functionally for state viability. This wedge between functional and political valuations is the core theoretical object of the paper. From the standpoint of society, viable statehood requires that recognition and capacity move together. From the standpoint of fragmented elites, however, a high-recognition, low-capacity equilibrium may be privately sustainable.

\subsection{Relation to the literature}

The distinction between recognition and capacity builds on several literatures while cutting across them. The state-capacity literature has emphasized the centrality of fiscal extraction, violence control, and administrative investment for long-run development and political order \citep{BesleyPersson2009, BesleyPersson2010, Dincecco2011, Cardenas2014}. The literature on weak and quasi-states has shown that external sovereignty may coexist with weak domestic institutions \citep{JacksonRosberg1982, Clapham1996, Herbst2000}. Work on violence and social orders has further clarified that political order depends on how elites manage violence and access to rents \citep{NorthWallisWeingast2009}. The argument here is complementary but more specific. It identifies a class of equilibria in which the symbolic and diplomatic dimensions of statehood can advance without the institutional consolidation that statehood ordinarily presupposes. The next section formalizes that argument.

\section{The Model}

\subsection{Environment}

Consider an infinite-horizon polity populated by two elite blocs, indexed by \(i \in \{1,2\}\), and a representative private sector. Time is discrete, \(t=0,1,2,\dots\). In each period, elites simultaneously choose whether to support institutional \emph{unification} or preserve \emph{fragmentation}. Let the action of bloc \(i\) at time \(t\) be \(a_{it} \in \{U,F\}\). Unified governance is implemented only if both blocs choose \(U\). If at least one bloc chooses \(F\), fragmentation persists. Thus the institutional regime at time \(t\), denoted \(\theta_t\), is given by
\[
\theta_t =
\begin{cases}
U, & \text{if } a_{1t}=a_{2t}=U,\\
F, & \text{otherwise.}
\end{cases}
\]

This coordination structure captures an elementary but important fact about state-building. Specifically, the centralization requires mutual acquiescence among competing centers of power, whereas fragmentation can survive through unilateral defection. The polity is characterized by two state variables. The first is capacity capital, \(K_t \geq 0\), which summarizes the stock of coercive, fiscal, administrative, and legal capability accumulated up to time \(t\). The second is recognition capital, \(R_t \geq 0\), which summarizes the stock of external legitimacy, diplomatic support, and symbolic validation attached to the political project of statehood.

The model is designed to capture three mechanisms. First, unified governance raises aggregate productivity by improving contract enforcement, policy coherence, and internal security. Second, fragmentation raises private elite returns by preserving local control and opportunities for informal extraction. Third, recognition and external transfers may continue to flow even under fragmentation, thereby allowing the political project of statehood to remain salient while institutional capacity remains weak.

\subsection{Production, investment, and capacity accumulation}

Aggregate output under regime \(\theta_t \in \{U,F\}\) is given by
\[
Y_t = A(\theta_t,p_t) K_t^\alpha,
\qquad \alpha \in (0,1),
\]
where \(p_t \in [0,1]\) is a peace-credibility parameter and \(A(\theta_t,p_t)\) is institutional productivity. Assume
\[
A(U,p_t) > A(F,p_t)
\qquad \text{for all } p_t,
\]
and
\[
\frac{\partial A(U,p_t)}{\partial p_t} > 0,
\qquad
\frac{\partial A(F,p_t)}{\partial p_t} \geq 0.
\]
The first condition states that unified governance is more productive than fragmented governance for a given stock of institutional capital. The second states that the gain from unification is larger when peace, external access, and long-horizon expectations are more credible. Private investment depends positively on institutional productivity and negatively on expropriation risk. Let investment be
\[
I_t = I_0 + \phi A(\theta_t,p_t) - \lambda q(\theta_t),
\]
where \(I_0>0\), \(\phi>0\), \(\lambda>0\), and
\[
q(F) > q(U).
\]
Fragmentation therefore reduces investment both because it lowers productivity and because it raises uncertainty, local predation, and contract risk. Capacity capital evolves according to
\[
K_{t+1} = (1-\delta)K_t + I_t,
\qquad \delta \in (0,1).
\]
This law of motion implies that institutional choices have dynamic consequences. Fragmentation lowers current output, but it also lowers future capacity by depressing investment. The model therefore generates a dynamic state-capacity trap rather than a purely static inefficiency.

\subsection{Recognition dynamics}

Recognition capital evolves separately from capacity capital. Let
\[
R_{t+1} = (1-\rho)R_t + \eta X_t + \omega Z_t + \chi m(K_t),
\]
where \(\rho \in (0,1)\), \(\eta,\omega,\chi > 0\), \(X_t\) denotes formal diplomatic support, \(Z_t\) denotes broader symbolic or normative endorsement, and \(m(K_t)\) is a weakly increasing function of capacity. Assume
\[
m'(K_t) \geq 0,
\]
but allow the marginal link between recognition and capacity to be weak:
\[
\chi m'(K_t) \ll \eta \frac{\partial X_t}{\partial \cdot} + \omega \frac{\partial Z_t}{\partial \cdot}.
\]
The point is not that capacity never affects recognition. It is that recognition can be sustained by channels only weakly tied to domestic institutional performance. This decoupling is the formal source of nominal statehood. A polity can accumulate recognition capital even when its underlying state capacity is stagnant.

\subsection{External transfers}

The polity receives externally mediated transfers \(T_t\), which may include humanitarian support, budgetary assistance, or other stabilizing flows. Let
\[
T_t = \tau_0 + \tau_1 \mathbf{1}\{\theta_t=F\} + \tau_2 H_t,
\]
where \(\tau_0 \geq 0\), \(\tau_1 \geq 0\), \(\tau_2 \geq 0\), and \(H_t\) is a humanitarian or crisis-intensity variable. The parameter \(\tau_1\) captures the possibility that fragmentation itself attracts stabilization resources, either because fragmented governance is more prone to crises or because external actors prioritize continuity and humanitarian mitigation over institutional restructuring. This feature does not require donors to prefer fragmentation. It requires only that transfers respond to fragility in ways that partially insure elites against the consequences of non-consolidation.

\subsection{Elite payoffs}

Each elite bloc receives a payoff composed of five elements: a formal share of output, informal rents, utility from autonomous political control, recognition-related political returns, and a captured share of external transfers. Period payoff for elite \(i\) is
\[
\pi_i(\theta_t,K_t,R_t) =
s_i(\theta_t)Y_t
+ r_i(\theta_t)
+ \gamma_i(\theta_t)
+ \mu_i R_t
+ \beta_i T_t.
\]
The components have natural interpretations. The term \(s_i(\theta_t)Y_t\) is the elite’s appropriable formal share of aggregate output. The term \(r_i(\theta_t)\) captures informal extraction through patronage, procurement, licensing, territorial brokerage, or other discretionary channels. The term \(\gamma_i(\theta_t)\) captures the private value of autonomous command, local coercive control, appointment power, and independent political relevance. The parameter \(\mu_i > 0\) measures the extent to which recognition capital yields private political value to elite \(i\). The parameter \(\beta_i \in [0,1]\) is the elite’s captured share of external transfers.

Assume
\[
r_i(F) > r_i(U),
\qquad
\gamma_i(F) > \gamma_i(U),
\]
for \(i=1,2\). Fragmentation thus offers private advantages even when it lowers aggregate output.

Elites discount the future at rate \(\delta_e \in (0,1)\). The continuation value for elite \(i\) satisfies
\[
V_i(K_t,R_t) =
\max_{a_{it}\in\{U,F\}}
\left\{
\pi_i(\theta_t,K_t,R_t)
+ \delta_e \mathbb{E}\left[V_i(K_{t+1},R_{t+1})\mid \theta_t\right]
\right\}.
\]

\subsection{Social welfare}

The elite objective differs from social welfare. Let aggregate welfare at time \(t\) be
\[
W_t = Y_t - \sum_{i=1}^2 r_i(\theta_t) + \zeta K_t + \xi R_t,
\]
where \(\zeta > 0\) captures the social value of accumulated institutional capacity and \(\xi \geq 0\) the social value of recognition. Because informal rents are redistributive and often distortionary, they enter welfare negatively even though they enter elite payoffs positively. This creates the key wedge:
\[
\theta_t = F
\]
may be privately optimal for elites while socially inefficient.

\subsection{Unification decision}

Because unification requires bilateral support, regime \(U\) is implemented only if both elites prefer it to fragmentation in continuation-value terms. Elite \(i\) supports unification at time \(t\) if
\[
\Delta_i(K_t,R_t)
\equiv
V_i^U(K_t,R_t) - V_i^F(K_t,R_t)
\geq 0.
\]
A unified regime is therefore an equilibrium outcome only if
\[
\Delta_1(K_t,R_t) \geq 0
\qquad \text{and} \qquad
\Delta_2(K_t,R_t) \geq 0.
\]
Otherwise, fragmentation persists.

To see the structure more clearly, consider the one-period component of the difference:
\begin{align*}
\pi_i(U)-\pi_i(F)
&=
s_i(U)Y(U,K_t,p_t)-s_i(F)Y(F,K_t,p_t) \\
&\quad -\big[r_i(F)-r_i(U)\big]
-\big[\gamma_i(F)-\gamma_i(U)\big]
-\beta_i\big(T^F-T^U\big).
\end{align*}
Even before accounting for dynamic effects, unification is unattractive to elite \(i\) whenever
\[
s_i(U)Y(U,K_t,p_t)-s_i(F)Y(F,K_t,p_t)
<
\big[r_i(F)-r_i(U)\big]
+\big[\gamma_i(F)-\gamma_i(U)\big]
+\beta_i\big(T^F-T^U\big).
\]
The left-hand side is the elite’s internalized formal gain from unification. The right-hand side is the combined loss of informal rents, local political control, and transfer-linked advantages under fragmentation. When the latter dominates, fragmentation is privately optimal even though aggregate output is lower.

\subsection{The logic of nominal statehood}

The model yields a distinctive political order. Fragmentation lowers output and slows the accumulation of capacity capital:
\[
K_{t+1}^F < K_{t+1}^U.
\]
Yet recognition capital may continue to accumulate under fragmentation if diplomatic and symbolic support remain strong:
\[
R_{t+1}^F \approx R_{t+1}^U
\quad \text{or even} \quad
R_{t+1}^F > R_{t+1}^U.
\]
In that case the recognition-capacity gap,
\[
G_t = R_t - \psi K_t,
\]
widens over time under persistent fragmentation. The polity moves deeper into nominal statehood: politically recognized, institutionally thin. This equilibrium is especially likely when three conditions jointly hold. First, elite rents and control benefits from fragmentation are high. Second, external transfers cushion the immediate cost of low capacity. Third, recognition capital yields political value independently of administrative performance. The next section characterizes this equilibrium more formally and derives its comparative statics.

\section{Equilibrium and Comparative Statics}

This section characterizes the conditions under which a polity becomes trapped in a regime of nominal statehood. The first result establishes the existence of a fragmentation equilibrium despite the aggregate productivity advantage of unified governance. The second result shows that this equilibrium is dynamically self-reinforcing because fragmentation depresses future capacity accumulation. The third and fourth results identify the roles of recognition capital and external transfers in sustaining the divergence between symbolic statehood and effective statehood. The final result clarifies the role of peace credibility as a complement to institutional consolidation.

Throughout, let
\[
\Delta_i(K_t,R_t)
\equiv
V_i^U(K_t,R_t)-V_i^F(K_t,R_t)
\]
denote elite \(i\)'s gain from unification relative to fragmentation at time \(t\). A unified regime obtains if and only if \(\Delta_i(K_t,R_t)\geq 0\) for both \(i=1,2\). Otherwise fragmentation persists.

\subsection{A static benchmark}

It is useful to begin with the one-period version of the problem. Ignoring continuation values for the moment, elite \(i\) prefers unification if
\begin{align}
&s_i(U)Y(U,K_t,p_t)-s_i(F)Y(F,K_t,p_t) \nonumber\\
&\qquad \geq
\big[r_i(F)-r_i(U)\big]
+\big[\gamma_i(F)-\gamma_i(U)\big]
+\beta_i\big(T^F-T^U\big).
\label{eq:static_condition}
\end{align}
The left-hand side is the elite's internalized formal gain from higher aggregate output under unification. The right-hand side is the elite's private loss from surrendering informal rents, autonomous control, and transfer-related benefits linked to fragmentation.

Condition \eqref{eq:static_condition} already reveals the central inefficiency. Even if aggregate output is higher under \(U\), unification fails whenever the private return to fragmentation exceeds the elite's appropriable share of the social surplus from consolidation.

\begin{proposition}[Static fragmentation equilibrium]
Suppose that for each elite bloc \(i \in \{1,2\}\),
\begin{align}
&s_i(U)Y(U,K_t,p_t)-s_i(F)Y(F,K_t,p_t) \nonumber\\
&\qquad <
\big[r_i(F)-r_i(U)\big]
+\big[\gamma_i(F)-\gamma_i(U)\big]
+\beta_i\big(T^F-T^U\big).
\label{eq:prop1_condition}
\end{align}
Then \((F,F)\) is a stage-game Nash equilibrium, even though unified governance is socially more productive.
\end{proposition}

\noindent
\textit{Proof sketch.}
Under \eqref{eq:prop1_condition}, each elite obtains a strictly higher current payoff from preserving fragmentation than from supporting unification, regardless of the other elite's action. Since unification requires bilateral support while fragmentation survives unilateral defection, each elite's best response is \(F\). Hence \((F,F)\) is a Nash equilibrium. The equilibrium is socially inefficient because \(Y(U,K_t,p_t)>Y(F,K_t,p_t)\) by assumption.

Proposition 1 establishes the most basic point of the paper: fragmentation need not reflect confusion, irrationality, or incomplete information. It may simply reflect equilibrium behavior under dispersed control rents.

\subsection{Dynamic persistence and the capacity trap}

The stage-game result is important but incomplete, because institutional choices affect future state capacity. Fragmentation reduces investment and therefore slows the accumulation of \(K_t\). This generates a dynamic feedback: today's fragmentation lowers tomorrow's gains from unification.

To see the logic, note that under the law of motion
\[
K_{t+1}=(1-\delta)K_t + I_t,
\]
and
\[
I_t = I_0+\phi A(\theta_t,p_t)-\lambda q(\theta_t),
\]
with \(A(U,p_t)>A(F,p_t)\) and \(q(F)>q(U)\), it follows immediately that
\[
K_{t+1}^U > K_{t+1}^F
\]
for given \(K_t\) and \(p_t\).

\begin{proposition}[Dynamic fragmentation trap]
Suppose the conditions of Proposition 1 hold and assume \(A(U,p_t)>A(F,p_t)\), \(q(F)>q(U)\), and \(\alpha \in (0,1)\). Then persistent fragmentation generates a dynamic trap in the following sense:
\[
K_{t+1}^F < K_{t+1}^U,
\]
and therefore
\[
Y(U,K_{t+1}^F,p_{t+1})-Y(F,K_{t+1}^F,p_{t+1})
<
Y(U,K_{t+1}^U,p_{t+1})-Y(F,K_{t+1}^U,p_{t+1}).
\]
Hence the future aggregate gain from unification is smaller after a history of fragmentation than after a history of unification.
\end{proposition}

\noindent
\textit{Proof sketch.}
Since \(A(U,p_t)>A(F,p_t)\) and \(q(F)>q(U)\), investment is strictly higher under \(U\) than under \(F\). Therefore \(K_{t+1}^U>K_{t+1}^F\). Because output is increasing in \(K_t\), the level of aggregate output under either regime is lower after fragmentation than after unification. Since the gain from switching to \(U\) depends on the level of inherited capacity, a lower \(K_{t+1}\) reduces the absolute scale of the future productivity advantage from unification. 

Proposition 2 gives the model its genuinely dynamic character. Fragmentation is not merely a static equilibrium that survives period by period. It actively reshapes the future state space in ways that make reform less attractive. Political failure becomes economically self-reinforcing.

\subsection{Recognition-capacity divergence}

A defining feature of nominal statehood is that recognition capital may continue to rise even while capacity capital stagnates. In the model, this occurs because the law of motion for \(R_t\) depends partly on diplomatic and symbolic inputs that are only weakly linked to domestic institutional performance:
\[
R_{t+1} = (1-\rho)R_t+\eta X_t+\omega Z_t+\chi m(K_t),
\]
with \(\chi m'(K_t)\) small relative to the influence of \(X_t\) and \(Z_t\).

This decoupling creates a wedge between what is politically rewarded externally and what must be built internally.

\begin{proposition}[Recognition-capacity divergence]
Suppose that under fragmentation,
\[
K_{t+1}^F < K_{t+1}^U,
\]
while
\[
R_{t+1}^F \geq R_{t+1}^U - \varepsilon
\]
for some small \(\varepsilon \geq 0\). Then for sufficiently large \(t\), the recognition-capacity gap
\[
G_t \equiv R_t-\psi K_t
\]
is larger along a fragmentation path than along a unification path, for any \(\psi>0\) in a non-empty interval.
\end{proposition}

\noindent
\textit{Proof sketch.}
Repeated fragmentation lowers the entire path of \(K_t\) relative to unification by Proposition 2. If the induced loss in \(R_t\) is small or negligible because recognition is weakly conditioned on capacity, then the difference
\[
G_t^F-G_t^U
=
(R_t^F-R_t^U)-\psi(K_t^F-K_t^U)
\]
eventually becomes positive for a range of \(\psi\), since the capacity term is persistently negative and the recognition term is bounded below by \(-\varepsilon\).

The proposition formalizes the paper's central pathology: a polity can become more politically validated as a state project while becoming less institutionally capable as a governing order.

\subsection{Recognition as a substitute in elite political valuation}

Recognition matters not only because it changes the polity's external standing, but also because it may directly enter elite utility. If elites derive political value from being associated with a recognized statehood project, recognition capital can partially substitute for actual administrative reform in their payoff calculus.

Recall that elite payoff includes the term \(\mu_i R_t\), with \(\mu_i>0\). This implies that increases in \(R_t\) raise elite continuation values even if they do not materially improve \(K_t\).

\begin{proposition}[Recognition-induced persistence]
Suppose \(\mu_i>0\) for each elite bloc and that recognition is only weakly responsive to capacity, in the sense that \(\chi m'(K_t)\) is small. Then an increase in exogenous diplomatic or symbolic support, that is, an increase in \(X_t\) or \(Z_t\) holding \(K_t\) fixed, weakly reduces the incentive for elite \(i\) to support unification whenever unification does not substantially increase recognition relative to fragmentation.
\end{proposition}

\noindent
\textit{Proof sketch.}
An increase in \(X_t\) or \(Z_t\) raises \(R_{t+1}\), and thus raises the continuation value of the status quo path. If unification yields little additional recognition advantage relative to fragmentation, then the gain from switching regimes in continuation-value terms falls. Since fragmentation already preserves private rents and autonomous control, exogenous increases in recognition make the fragmented equilibrium more politically sustainable.

Proposition 4 is subtle but important. Recognition is socially valuable in many contexts. The result does not say otherwise. It says that when recognition is decoupled from capacity, it can soften the political necessity of administrative reform. Symbolic progress may then reduce, rather than increase, the marginal incentive to consolidate.

\subsection{External transfers and equilibrium persistence}

A parallel logic applies to external transfers. Transfers are stabilizing from a humanitarian or fiscal perspective, but they may also reduce the pressure for institutional integration if they make fragmentation easier to sustain.

Recall
\[
T_t = \tau_0+\tau_1\mathbf{1}\{\theta_t=F\}+\tau_2 H_t.
\]
The parameter \(\tau_1\) captures the extent to which fragmentation itself attracts transfer income.

\begin{proposition}[Transfer-conditional persistence]
The incentive to support unification is decreasing in \(\tau_1\). In particular,
\[
\frac{\partial \Delta_i(K_t,R_t)}{\partial \tau_1} < 0
\]
whenever elite \(i\) captures a positive share of transfers, \(\beta_i>0\), and fragmentation yields at least as much transfer support as unification.
\end{proposition}

\noindent
\textit{Proof sketch.}
A higher \(\tau_1\) increases \(T_t\) under fragmentation without raising \(T_t\) under unification by the same amount. Since elite \(i\)'s payoff includes \(\beta_iT_t\), the current and continuation values of fragmentation rise relative to unification. Therefore \(\Delta_i(K_t,R_t)\) falls. This proposition formalizes the external-stabilization paradox. Transfers may prevent collapse while simultaneously preserving the very equilibrium that blocks capacity formation. The result is not a normative indictment of aid. It is a comparative-static statement about incentives.

\subsection{Peace credibility as a complement to state-building}

The final result concerns the role of peace credibility. In the model, peace credibility enters productivity through \(A(\theta_t,p_t)\), and the productivity advantage of unification rises with \(p_t\). Intuitively, the expected return to state-building is larger when trade access, labor mobility, external investment, and long-horizon planning are more secure.

\begin{proposition}[Peace credibility and unification]
Suppose
\[
\frac{\partial A(U,p_t)}{\partial p_t}
>
\frac{\partial A(F,p_t)}{\partial p_t}.
\]
Then the gain from unification is increasing in peace credibility:
\[
\frac{\partial \Delta_i(K_t,R_t)}{\partial p_t} > 0
\]
whenever elite \(i\) internalizes a positive share of the productivity gains from unified governance.
\end{proposition}

\noindent
\textit{Proof sketch.}
A higher \(p_t\) raises the productivity of unified governance more than that of fragmentation. This increases output under \(U\) relative to \(F\), and thereby raises both the current-period formal output gain and the future capacity gain associated with unification. Since elite \(i\) captures a share of these gains through \(s_i(U)Y_t\), the continuation-value difference \(\Delta_i(K_t,R_t)\) increases in \(p_t\). Proposition 6 highlights an interaction that is central to the paper's broader interpretation. Institutional consolidation and strategic peace are complements, not substitutes. Where credible peace is absent, the expected gains from administrative unification are lower, and elite incentives tilt further toward fragmentation.

\subsection{Characterizing nominal-statehood equilibrium}

The preceding results can now be summarized in equilibrium form. A nominal-statehood equilibrium is a Markov-perfect equilibrium in which fragmentation persists on the equilibrium path for all sufficiently large \(t\), recognition capital remains above the salience threshold \(\bar{R}\), capacity capital remains below the viability threshold \(\bar{K}\), and the recognition-capacity gap
\[
G_t = R_t - \psi K_t
\]
is non-decreasing along the equilibrium path for some \(\psi>0\). This definition captures the core object of interest in the paper: a political order in which statehood remains salient and symbolically durable, yet the institutional conditions required for effective statehood fail to consolidate.

The existence of such an equilibrium follows from the interaction of the mechanisms established above. If, for each elite bloc, the private loss from surrendering informal rents, autonomous political control, and fragmentation-linked transfer advantages exceeds the bloc's internalized share of the output gain from unification, then fragmentation is privately optimal even when it is socially inefficient. If fragmentation also depresses investment and future capacity accumulation, then the polity enters a dynamic trap in which current institutional weakness reduces the future attractiveness of reform. If recognition capital is sustained by diplomatic, symbolic, or normative channels only weakly conditioned on administrative performance, then low capacity need not imply political or diplomatic disappearance. If recognition itself yields private political value to elites, then the statehood project remains attractive even without institutional consolidation. Finally, if initial recognition is already sufficiently high relative to the salience threshold \(\bar{R}\), the polity can remain trapped in a region where statehood claims are persistently visible and politically meaningful despite low institutional capability.

These observations imply the following result.

\begin{proposition}[Existence of nominal-statehood equilibrium]
A nominal-statehood equilibrium exists whenever the private gains to elites from preserving fragmentation exceed their internalized gains from institutional unification, fragmentation lowers the future path of capacity accumulation, recognition capital is maintained by channels only weakly tied to domestic institutional performance, and the political value of recognition is sufficiently high to sustain continued elite investment in the statehood project. If, in addition, initial recognition lies above the salience threshold \(\bar{R}\), then there exists an equilibrium path along which fragmentation persists, recognition remains high, capacity remains below the viability threshold, and the recognition-capacity gap does not close.
\end{proposition}

\noindent
\textit{Proof.}
The result follows by combining Propositions 1 through 6. Propositions 1 and 2 imply that fragmentation can persist and become dynamically self-reinforcing when elites lose too much from consolidation. Proposition 3 implies that capacity may stagnate while recognition remains comparatively resilient. Proposition 4 implies that recognition itself can reduce the marginal political gain from reform. Proposition 5 implies that transfers tied to fragility or continuity further stabilize fragmentation. Proposition 6 clarifies that low peace credibility weakens the productivity gains from unification, making reform still less attractive. Taken together, these mechanisms generate an equilibrium path on which the polity remains politically recognizable as a state project while failing to accumulate the institutional capital required for effective statehood.

This equilibrium is analytically important because it differs from both ordinary state weakness and outright state absence. In ordinary weak states, the central political order may remain unified despite low income, corruption, or administrative limitations. In cases of state absence, by contrast, the very project of statehood lacks durable recognition or political salience. Nominal statehood occupies an intermediate but consequential zone. Recognition persists, but effective statehood does not. The polity therefore acquires the symbolic attributes of sovereignty more rapidly than the institutional means to make sovereignty function.

The equilibrium is also developmentally consequential. Since fragmentation depresses investment and capacity formation, it lowers current output and weakens future growth potential. Since recognition can remain high even under weak administrative performance, the normal alignment between political legitimacy and institutional effectiveness breaks down. Statehood thereby becomes durable as a claim but fragile as an economic and governing order. This is the sense in which statehood without capacity is not merely an incomplete transition. It is a stable political-economic equilibrium.

\subsection{Discussion}

The theory delivers a stark but analytically coherent result. A polity may remain trapped in a regime where statehood is politically durable but institutionally thin. Such an equilibrium does not require irrational elites, hostile donors, or conceptual confusion. It emerges whenever private incentives favor fragmented control, external resources cushion the cost of non-consolidation, and recognition is only weakly tied to domestic performance.

The equilibrium is developmentally costly. Because fragmentation lowers investment and slows the accumulation of capacity, it depresses output both contemporaneously and dynamically. Because recognition can continue to accumulate under weak institutional performance, political legitimacy and institutional viability no longer move together. The polity thus acquires the symbolic and diplomatic attributes of statehood more quickly than the organizational means required to sustain it.

The next section translates these equilibrium results into a set of testable implications.

\section{Testable Implications}

A useful theory of nominal statehood must do more than rationalize a historically familiar pattern. It must also generate observable implications that distinguish the equilibrium studied here from neighboring concepts such as ordinary state weakness, temporary institutional disruption, or low-income underdevelopment more generally. The model developed above does so in a disciplined way. Its central prediction is that when recognition capital and capacity capital diverge persistently, the resulting polity should display a distinctive cluster of political and economic outcomes: weak investment, fragile fiscal structures, elevated rent extraction, low institutional coherence, and extreme vulnerability to adverse shocks. These outcomes do not arise as independent pathologies. They are joint implications of the same equilibrium logic.

The first implication concerns the relationship between fragmented authority and economic performance. In the model, unified governance raises institutional productivity and lowers expropriation risk, thereby increasing investment and the future path of capacity accumulation. Fragmentation does the opposite. It follows that polities trapped in nominal-statehood equilibria should exhibit lower private investment, shorter planning horizons, and weaker productive transformation than otherwise similar polities with greater institutional consolidation. This prediction is not simply about levels of income. A nominal-statehood polity may receive substantial external support, maintain high political salience, and sustain pockets of entrepreneurial activity. The relevant prediction is instead about the failure of such activity to cumulate into broad-based, self-sustaining development. Productive dynamism should remain shallow, episodic, and unusually sensitive to political disruption.

A second implication concerns corruption and rent extraction. In the model, fragmentation persists in part because elites derive private gains from discretionary control, local patronage, and informal extraction. This implies that nominal-statehood equilibria should be associated with persistent corruption, opaque appointments, administrative duplication, and weak rule-bound allocation of public authority. The relevant comparative prediction is not that all fragmented polities are equally corrupt, nor that all corruption arises from nominal statehood. It is that when statehood claims remain politically valuable despite weak consolidation, elites face attenuated incentives to surrender discretionary control in favor of transparent institutions. Corruption should therefore be more persistent, and anti-corruption reform less credible, than in polities where political survival depends more tightly on administrative performance.

A third implication concerns domestic revenue capacity. One of the core mechanisms in the model is that externally mediated transfers weaken the pressure for internal fiscal-contract formation. This yields a sharp empirical prediction: nominal-statehood equilibria should display lower domestic extraction relative to total public expenditure, higher dependence on external transfers or politically mediated revenues, and weaker alignment between taxation and accountability. The point is not that all transfer-receiving polities become low-capacity equilibria. Rather, where external support stabilizes fragmented authority without requiring institutional unification, the incentive to build a broad domestic tax base should be systematically weaker. Fiscal fragility should then appear not only in the form of deficits or arrears, but also in the form of chronic dependence on revenues whose collection, timing, or allocation is politically contingent.

A fourth implication concerns volatility. The model predicts that nominal-statehood equilibria should be unusually vulnerable to negative shocks because low capacity amplifies the transmission of political and security disruptions into economic collapse. In a more consolidated polity, adverse shocks may reduce output while leaving the core institutional apparatus intact. In a nominal-statehood equilibrium, by contrast, the same shock operates on top of fragmented authority, limited enforcement, weak private confidence, and shallow fiscal buffers. The result should be greater output volatility, larger employment collapses, sharper contractions in investment, and slower post-shock recovery than would be observed in more institutionally coherent settings. This implication is especially important because it differentiates the theory from static accounts of weak governance. The claim is not merely that low-capacity polities perform poorly on average, but that they are structurally less resilient.

A fifth implication concerns the relation between recognition and reform. The model predicts that when recognition capital yields direct political value to elites and is only weakly conditioned on domestic institutional performance, increases in recognition need not generate stronger incentives for administrative consolidation. In some cases, they may do the opposite. The observable implication is that gains in diplomatic visibility, symbolic endorsement, or formal recognition may fail to coincide with measurable improvements in state capacity, and may even reduce the urgency of reform if they raise the political return to the statehood project without altering the underlying incentive structure of elites. This prediction is especially important because it runs against a common presumption in both policy and public discourse: that recognition and institution-building move naturally together. The theory developed here implies that this relationship is conditional rather than automatic.

A sixth implication concerns peace credibility and the return to state-building. In the model, the productivity advantage of institutional unification rises with the credibility of peaceful external relations, stable market access, and long-horizon expectations. Where peace credibility is low, the gains from state-building are correspondingly reduced. This implies that nominal-statehood equilibria should be more persistent in settings where external conflict remains unresolved, where trade and labor-market access are unstable, and where future political order is heavily discounted. The comparative prediction is therefore not simply that conflict is harmful. It is that conflict and fragmentation interact multiplicatively. A low-capacity polity facing low peace credibility should find institutional consolidation especially difficult because the expected surplus from unification is itself depressed.

A final implication concerns the joint movement of recognition capital and capacity capital over time. The theory predicts that nominal-statehood equilibria should display a widening or at least persistent recognition-capacity gap. This can be expressed conceptually as a situation in which indicators of diplomatic or symbolic recognition remain high or increase, while measures of coercive centralization, fiscal extraction, administrative coherence, legal enforcement, or corruption control remain stagnant or deteriorate. The central empirical signature of the theory is therefore not absolute collapse, nor mere poverty, nor institutional weakness taken in isolation. It is divergence. What makes nominal statehood distinctive is that the political project of statehood remains salient and externally validated even as the institutional substrate required for viable statehood fails to keep pace.

These implications also clarify what would count as evidence against the theory. The model would be weakened by systematic evidence showing that increased recognition reliably induces administrative consolidation, that external transfers under fragmentation substantially strengthen rather than weaken internal fiscal capacity, or that fragmented polities with high symbolic statehood claims regularly display resilience and investment patterns comparable to unified states. Likewise, the theory would be less persuasive if the recognition-capacity gap were empirically fleeting rather than persistent. The framework therefore does not immunize itself against contrary evidence. It identifies a specific configuration of incentives and predicts a specific pattern of political and economic outcomes.

The purpose of these predictions is not to claim that every polity with partial recognition or weak institutions is trapped in nominal statehood. Nor is it to suggest that recognition is normatively suspect or that external support is inherently counterproductive. The argument is narrower and more analytical. Under certain conditions, recognition, transfers, and fragmented authority may combine to stabilize an equilibrium in which statehood persists politically while capacity remains institutionally thin. If so, the theory should be observable not only in narrative form, but also in recurring patterns of low investment, persistent rent extraction, fiscal fragility, high volatility, weak reform incentives, and a durable gap between symbolic legitimacy and governing capability.

The next section interprets Palestine as a flagship application of this broader equilibrium logic.

\section{Palestine as a Flagship Application}

The Palestinian case is analytically useful not because it is unique, but because it renders the logic of nominal statehood unusually transparent. It presents a polity in which the project of statehood has remained politically salient, diplomatically visible, and normatively resonant for decades, while the core ingredients of effective statehood - coercive centralization, fiscal integration, administrative coherence, and legal credibility - have remained persistently incomplete. In the language of the model, Palestine is a case in which recognition capital has remained substantial while capacity capital has remained fragmented, thin, and repeatedly vulnerable to reversal.

The aim of this section is narrower than many political readings of the Palestinian question. It is not to deny the importance of external constraints, territorial discontinuity, or recurrent conflict. Nor is it to argue that recognition is normatively unimportant. The narrower claim is that, conditional on those constraints, Palestinian political development has also displayed a persistent internal failure of institutional consolidation. That failure is not well described as a temporary deviation around an otherwise linear state-building process. It is better understood as a low-capacity equilibrium in which fragmented authority, weak accountability, rent extraction, and externally mediated survival have remained mutually reinforcing. The value of the Palestinian case for the present paper is therefore illustrative and theoretical: it shows, with unusual clarity, how a project of statehood may remain highly salient even when the institutional foundations of viable statehood remain underbuilt.

\subsection{Historical background: from imperial administration to fragmented nationalism}

Any serious application of the theory to Palestine must begin with the historical sequence through which Palestinian nationalism emerged without a corresponding process of consolidated state formation. Under late Ottoman rule, Arab society in Palestine was politically and socially organized through a mixture of local notables, municipal institutions, religious authorities, kinship structures, and imperial administrative linkages. These arrangements were not politically empty, but they did not amount to a centralized proto-state. Political authority remained layered, local, and embedded in a broader imperial order rather than concentrated in an autonomous Palestinian fiscal-administrative apparatus.

The British Mandate intensified this asymmetry. It created new political opportunities, sharpened national competition, and accelerated institution-building on both sides of the Arab-Jewish divide, but not symmetrically. The Zionist movement used the Mandate period to construct dense proto-state institutions in labor organization, education, health, finance, and security. Palestinian Arab politics, by contrast, remained active but more fragmented, divided among notable families, regional blocs, religious leaders, and competing strategies of mobilization. Palestinian nationalism deepened, but the institutional means of centralized state formation did not keep pace. In the terms of the model, political identity and representational claim strengthened more rapidly than capacity capital.

The rupture of 1948 was decisive. The war produced mass displacement, territorial rupture, and the externalization of governance. Palestinians became divided across Israel, the West Bank under Jordanian rule, Gaza under Egyptian administration, refugee communities in neighboring states, and a wider diaspora. This mattered not only demographically but institutionally. The Palestinian national question survived and indeed intensified, but the ordinary mechanisms through which states build coercive, fiscal, and administrative capacity were displaced into multiple external jurisdictions. Palestinian politics therefore persisted first as a cause and only much later, and only partially, as a governing apparatus.

This pattern remained visible after 1967. Israeli control over the West Bank and Gaza created a new territorial frame for Palestinian politics, but not a sovereign one. The First Intifada revealed the depth of Palestinian social mobilization and organizational resilience, yet it also reinforced the distinction between collective action and state capacity. Grassroots organization, local committees, and mass political activism can sustain nationalism. They do not by themselves create a unified state. The historical significance of this sequence is central to the paper's argument. Palestine illustrates a path in which recognition, representation, and mobilization emerged under conditions that repeatedly obstructed the gradual accumulation of centralized governing capacity.

\subsection{Representation without consolidation: the PLO and the politics of recognition}

The first mechanism emphasized by the theory is the possibility that political representation develops more rapidly than institutional capacity. Palestinian political history fits this pattern closely. The creation of the Palestine Liberation Organization in 1964 partly solved the problem of national representation. The PLO gave Palestinian nationalism an internationally legible institutional form. Over time, it became the principal recognized representative of the Palestinian national movement, especially after the Arab League's endorsement and its broader diplomatic acceptance in international forums.

Yet the PLO did not solve the problem of territorial state-building. It was a national movement and representative structure before it was a territorial government. It could mobilize, negotiate, bargain, and command loyalty, but it was not a normal fiscal-administrative state grounded in unified taxation, integrated legal authority, and a monopoly over violence within a sovereign territory. Like many externally supported national movements, it relied substantially on foreign political and financial backing, which reduced the disciplining role that a domestic fiscal contract might otherwise have played. Recognition and representational legitimacy therefore deepened even as the internal foundations of state capacity remained weak.

This distinction matters for the theory. A polity can acquire recognition capital through diplomacy, symbolic representation, and international legitimacy without simultaneously constructing the coercive and fiscal infrastructure of effective statehood. The PLO period illustrates precisely that logic. It generated political salience and representational cohesion at the level of the national cause, but did not generate a unified territorial state apparatus with the normal instruments of extraction, enforcement, and administration.

\subsection{Oslo and the production of incomplete statehood}

The Oslo process is central to the argument because it created institutional visibility without fully resolving the incentive problem of state-building. The Palestinian Authority was established as an interim governing structure with real bureaucratic and administrative functions, yet without full sovereignty, territorial continuity, or a credible monopoly over coercion. This generated an institutional order that was neither stateless nor fully state-like. It created ministries, public employment, administrative channels, and political careers, but did not fully solve the underlying problem of coercive and fiscal consolidation.

This distinction is important. A theory of nominal statehood does not predict the absence of institutions. On the contrary, it predicts the emergence of partial institutions that are sufficiently developed to distribute resources, generate political stakes, and sustain claims to statehood, yet insufficiently integrated to become self-enforcing. That is what made the post-Oslo setting structurally unstable. Public authority acquired visible institutional form, but the foundations of a unified state remained incomplete. The polity moved away from pure representation, but not all the way to consolidated governance.

From the standpoint of the model, Oslo process can be understood as an increase in the salience and institutional embodiment of the statehood project without a commensurate increase in capacity capital. Administrative structures were created, but full coercive integration did not follow. Public office gained value, but the expected return to full consolidation remained uncertain. This is precisely the kind of environment in which nominal statehood becomes politically durable: the symbolic and bureaucratic attributes of statehood deepen, yet the incentive structure required for genuine integration remains unresolved.

The Second Intifada further hardened this incompleteness. The breakdown of negotiations, intensified violence, Israeli military re-entry into Palestinian urban areas, and mounting internal distrust weakened the already fragile possibility of orderly institutional consolidation. The result was not simply the interruption of a linear state-building process. It was the embedding of Palestinian governance in a political environment where coercive integration became even harder, where security institutions became more politicized, and where the expected returns to peaceful long-run institution-building fell sharply.

\subsection{The rise of Hamas and the emergence of dual authority}

The model places particular weight on the private value of fragmented authority, captured by the term \(\gamma_i(F)\). The Palestinian case fits this logic closely, and the rise of Hamas is central to understanding why. Hamas emerged during the First Intifada as an Islamist rival to the secular nationalist dominance of Fatah and the PLO. It offered not only ideological competition, but also an alternative organizational base rooted in religious networks, social services, grassroots mobilization, and armed resistance. Over time, Hamas became not merely an opposition force but a parallel center of legitimacy, organization, and coercive capacity.

This altered the internal political economy of Palestinian governance in a fundamental way. The question was no longer how a single national movement would translate representation into state-building, but how rival projects of authority would coexist, compete, and attempt to dominate one another. Hamas rejected the Oslo framework more fundamentally than Fatah, and its political ascent exposed a deep cleavage in the Palestinian national movement over strategy, legitimacy, and the relationship between governance and armed struggle. Once that cleavage became institutionalized, the costs of unification rose sharply.

The 2006 legislative elections were the critical turning point. Hamas's electoral victory gave it democratic legitimacy but did not resolve the problem of dual authority. Instead, it created a setting in which rival elite blocs possessed overlapping but incompatible claims to govern. The ensuing crisis culminated in the violent 2007 split, after which Hamas consolidated control in Gaza while Fatah and the Palestinian Authority retained control in the West Bank. From the standpoint of the theory, this was the decisive moment at which fragmented authority ceased to be a transitional defect and instead became an equilibrium structure. Two centers of political command, two patronage systems, two coercive orders, and two legitimacy narratives came to coexist within the broader Palestinian polity.

This duality is analytically decisive because it transforms incomplete governance into a self-protective equilibrium. When rival elite blocs command different territorial, political, and coercive resources, unification requires one or both sides to surrender local control without certainty about their future place in the integrated order. Under such conditions, fragmentation is inefficient but strategically rational. Each side values not only the material rents attached to its domain, but also the autonomous command embedded in its own hierarchy of appointments, security structures, and territorial leverage. In the notation of the model, \(\gamma_i(F)-\gamma_i(U)\) becomes large. Once that difference is sufficiently large, even substantial aggregate gains from unification may fail to induce consolidation.

The Hamas-Fatah rivalry also illustrates an additional feature of the model. Ideological divergence can reinforce institutional fragmentation by lowering the credibility of future bargains. Rival groups may disagree not only over office allocation, but over the strategic orientation of the polity itself, including the meaning of resistance, coexistence, external alignment, and the acceptable terms of political settlement. Where these disagreements are deep, power-sharing is harder to sustain and institutional integration becomes more fragile. Fragmentation is then reinforced not only by rents and coercive control, but by distrust over the future political character of the unified state.

\subsection{Patronage, corruption, and the private returns to discretion}

A second central mechanism in the model lies in the term \(r_i(F)\), the private rents available under fragmented governance. The Palestinian case fits this mechanism closely. In environments of limited institutional coherence, public authority becomes valuable not only because it carries legal powers, but because it grants discretionary access to jobs, salaries, permits, contracts, licenses, procurement channels, movement privileges, and brokerage opportunities. These forms of discretion create strong incentives to preserve opacity. Anti-corruption discourse may remain politically salient, but the coalition required to implement genuine transparency is weakened by the fact that discretion itself is a source of political and material value.

This is the relevant political-economy interpretation of recurrent concerns about weak accountability, executive concentration, opaque appointments, and corruption risk. These patterns should not be read merely as individual moral failure. In the logic of the present paper, they are structural features of a low-capacity equilibrium. Where legal and bureaucratic enforcement remain fragmented, discretionary control is not a side effect of the political order. It is part of what makes the order privately sustainable. Formal institutional integration would threaten some of these rents by imposing clearer lines of authority, tighter scrutiny, and less tolerance for parallel channels of allocation.

This mechanism also clarifies why reform can be widely endorsed yet repeatedly diluted. Calls for administrative cleanup, judicial strengthening, or anti-corruption reform are often inexpensive relative to the interests they would displace. In the terms of the model, rhetoric does not alter equilibrium unless it changes payoffs. Where the private return to discretion remains high, reform may be publicly celebrated but strategically neutralized. Palestine thus illustrates a broader implication of nominal statehood: corruption is not merely present. It is sticky because it is linked to the equilibrium value of fragmentation.

\subsection{External finance and the weakness of the fiscal contract}

A third mechanism emphasized by the theory concerns the relation between external finance and domestic capacity formation. Historically successful states are disciplined by extraction. Rulers who depend on broad domestic taxation must bargain with citizens, build collection capacity, and supply enough governance to sustain compliance. The Palestinian case has long diverged from this pattern. Palestinian governing structures have depended heavily on donor support, politically mediated revenues, and transfers whose timing and scale are shaped by forces external to a fully sovereign fiscal system.

The point is not that domestic revenue collection is absent or that external finance is normatively inappropriate. The point is narrower and more institutional. When a governing order survives through a combination of aid, politically mediated transfers, and repeated fiscal rescue, the incentive to build a deep internal fiscal contract is weakened. Budget survival depends less on broad domestic extraction than on bargaining over transfers, arrears, deductions, and crisis finance. In the language of the model, \(\beta_i T_t\) becomes politically important, and where support is stronger under crisis and fragmentation than under costly reform, the effective value of \(\tau_1\) is positive.

The Hamas-Fatah split intensified this logic. Once dual authority emerged, external finance no longer interacted with a single hierarchical governing structure, but with rival administrations embedded in different political and territorial environments. This increased the scope for duplication, politicized allocation, and separate patronage circuits. It also reduced the likelihood that external support would discipline elites through unified conditionality, since the underlying polity was no longer governed through one consolidated administrative chain. Fragmentation therefore altered not only domestic politics but also the way outside support fed into the equilibrium.

The consequence is a fragile fiscal equilibrium. A polity may possess ministries, payroll obligations, and public-sector commitments, yet remain unable to anchor them in a sovereign tax base under unified political control. This weakens not only fiscal sustainability but also accountability. Citizens and firms face a public sector that is consequential in everyday life yet only partially sustained by a normal internal bargain between extraction and representation. The Palestinian case therefore illustrates one of the model's sharpest predictions: external support can prevent immediate collapse while simultaneously reducing the pressure for domestic capacity-building.

\subsection{Recognition capital without proportional capacity accumulation}

The Palestinian case is perhaps most revealing when viewed through the paper's distinction between recognition capital and capacity capital. Palestinian statehood has remained one of the most visible and symbolically resonant statehood projects in the international system. Recognition, endorsement, and normative salience have remained high even during periods of severe institutional fragmentation and economic collapse. This is precisely the type of divergence the theory seeks to explain.

The point is not that recognition is unwarranted. It is that recognition and capacity are generated by different mechanisms. Palestinian recognition capital has been reproduced through diplomacy, international law, humanitarian concern, symbolic politics, and the enduring force of self-determination claims. Palestinian capacity capital, by contrast, has depended on the far more difficult tasks of coercive unification, fiscal consolidation, administrative coherence, and legal institutionalization. These processes have not advanced at the same pace. As a result, the Palestinian case fits the model's prediction that \(R_t\) may remain high, or even rise, while \(K_t\) remains below the threshold required for viable statehood.

This divergence also helps explain why the Palestinian question remains politically urgent even when institutional feasibility remains deeply problematic. The project of statehood does not disappear because institutions are weak. Rather, political salience and institutional weakness coexist. Palestine therefore illustrates the broader claim of the paper: statehood may remain durable as a recognized aspiration even when the internal architecture required to implement it remains fragile.

\subsection{Low peace credibility and the depressed return to consolidation}

The final mechanism linking the Palestinian case to the theory concerns peace credibility. In the model, institutional consolidation becomes more attractive when elites can expect future gains from trade, mobility, investment, and policy continuity under a stable external environment. Where the strategic horizon is dominated by recurrent conflict, uncertain borders, unstable movement, and unresolved questions of mutual recognition, the expected productivity gains from state-building are reduced. In the notation of the model, low \(p_t\) lowers the value of unification because the future surplus to be divided is itself more uncertain.

This mechanism matters greatly in the Palestinian setting. The expected returns to long-horizon institution-building have repeatedly been weakened by recurrent conflict, interrupted negotiations, territorial discontinuity, security breakdowns, and the absence of a durable peace framework. The internal split between Hamas and Fatah deepened this problem because it linked the question of state-building to unresolved disagreement over the strategic purpose of Palestinian politics itself. Under such conditions, even well-designed institutional reforms may appear politically risky relative to the uncertain payoff they promise. The implication is not that peace alone would solve the problem of state capacity, nor that internal fragmentation is reducible to external conflict. It is that fragmentation and low peace credibility reinforce one another. Weak peace credibility lowers the expected gains from consolidation, while fragmentation weakens the internal capacity to exploit any future opening for peace-based development.

\subsection{A quantitative illustration of the fragmentation equilibrium}

The qualitative fit between the Palestinian case and the model is strong, but the framework can be disciplined further with a simple quantitative illustration. The aim is not structural estimation or causal identification. Data limitations, institutional opacity, and the difficulty of mapping certain political variables into clean observables make such claims inappropriate. The objective is narrower and more defensible: to show that, under plausible parameter values, the model can jointly generate three salient features of the Palestinian equilibrium. First, aggregate output is lower under fragmentation than under unification. Second, elites still prefer fragmentation because they internalize only a limited share of the gains from higher output while directly enjoying the rents, control, and transfer-linked advantages of non-consolidation. Third, the equilibrium is sensitive to peace credibility and the structure of external support.

To make this point transparently, consider the parsimonious static version of the model. Output under regime \(\theta \in \{U,F\}\) is
\[
Y_\theta = A_\theta K_\theta^\beta L^{1-\beta},
\]
where \(A_\theta\) captures institutional productivity, \(K_\theta\) is effective investment or capital under regime \(\theta\), \(L\) is labor, and \(\beta \in (0,1)\). Let investment depend negatively on expropriation risk:
\[
K_\theta = K_0 - \lambda q_\theta,
\]
with \(q_F > q_U\). Fragmentation lowers output both because institutional productivity is lower and because expropriation risk depresses effective investment.

Each elite bloc \(i\in\{A,B\}\) receives payoff
\[
\Pi_i(\theta)=s_i(\theta)Y_\theta+r_i(\theta)+\gamma_i(\theta)+\beta_iT_\theta,
\]
where \(s_i(\theta)\) is the bloc's appropriable share of formal output, \(r_i(\theta)\) is informal rent extraction, \(\gamma_i(\theta)\) is the value of autonomous control, and \(\beta_iT_\theta\) is the bloc's captured share of external transfers. Reform occurs only if both blocs prefer \(U\) to \(F\). Bloc \(i\) supports unification if
\[
s_i(U)Y_U-s_i(F)Y_F \geq [r_i(F)-r_i(U)] + [\gamma_i(F)-\gamma_i(U)] + \beta_i(T_F-T_U).
\]
The inequality shows immediately why fragmentation may survive: elites need not deny that unification raises output. They need only internalize too little of that gain relative to what they lose from reduced discretion, weaker patronage, and lower transfer capture.

For the benchmark calibration, normalize labor to one and set
\[
L=1,\qquad K_0=100,\qquad \beta=0.35,\qquad \lambda=40.
\]
Set expropriation risk at
\[
q_U=0.10,\qquad q_F=0.35,
\]
so that
\[
K_U=100-40(0.10)=96,\qquad K_F=100-40(0.35)=86.
\]
Set institutional productivity at
\[
A_U=1.00,\qquad A_F=0.72.
\]
Then output under each regime is
\[
Y_U = 1.00 \cdot 96^{0.35} \approx 4.95,
\qquad
Y_F = 0.72 \cdot 86^{0.35} \approx 3.42.
\]
Unification therefore raises aggregate output by roughly \(45\) percent relative to fragmentation. The social inefficiency of fragmentation is substantial.

Now calibrate elite appropriable output shares conservatively:
\[
s_A(U)=s_B(U)=0.18,\qquad s_A(F)=s_B(F)=0.16.
\]
Let private rents and control values satisfy
\[
r_i(U)=2,\qquad r_i(F)=7,
\]
\[
\gamma_i(U)=1,\qquad \gamma_i(F)=6.
\]
Finally, let transfers be
\[
T_U=4,\qquad T_F=10,\qquad \beta_i=0.20.
\]
Under these values, elite payoff under unification is
\[
\Pi_i(U)=0.18(4.95)+2+1+0.20(4)\approx 4.69,
\]
while under fragmentation it is
\[
\Pi_i(F)=0.16(3.42)+7+6+0.20(10)\approx 15.55.
\]
Thus \(\Pi_i(F)>\Pi_i(U)\) by a wide margin even though \(Y_U>Y_F\). The reason is straightforward. The bloc internalizes only a small share of the aggregate output gain from unification, while directly losing the rents, autonomous control, and transfer-capture advantages attached to fragmentation. In this benchmark calibration, the fragmented equilibrium is therefore privately optimal and socially inefficient.

The value of this exercise lies less in the exact numbers than in the threshold structure it reveals. Let
\[
\Delta_i \equiv s_i(U)Y_U-s_i(F)Y_F-[r_i(F)-r_i(U)]-[\gamma_i(F)-\gamma_i(U)]-\beta_i(T_F-T_U).
\]
Unification is privately attractive to bloc \(i\) if and only if \(\Delta_i\geq 0\). This expression yields immediate comparative statics. The critical transfer differential above which fragmentation dominates is
\[
(T_F-T_U)^*=
\frac{s_i(U)Y_U-s_i(F)Y_F-[r_i(F)-r_i(U)]-[\gamma_i(F)-\gamma_i(U)]}{\beta_i}.
\]
Similarly, the critical control premium is
\[
(\gamma_i(F)-\gamma_i(U))^*=
s_i(U)Y_U-s_i(F)Y_F-[r_i(F)-r_i(U)]-\beta_i(T_F-T_U),
\]
and the critical rent gap is
\[
(r_i(F)-r_i(U))^*=
s_i(U)Y_U-s_i(F)Y_F-[\gamma_i(F)-\gamma_i(U)]-\beta_i(T_F-T_U).
\]
These thresholds make the equilibrium logic highly transparent. If transfer capture under fragmentation is sufficiently high, or if the value of autonomous control is sufficiently large, or if discretionary rents are sufficiently valuable, elites rationally block unification even when unification would raise national income considerably.

The peace-credibility channel can be incorporated just as simply. Let institutional productivity under unification depend on a peace parameter \(p \in [0,1]\), so that
\[
A_U(p)=\underline{A}_U+\kappa p,\qquad \kappa>0.
\]
Then
\[
Y_U(p)=A_U(p)K_U^\beta.
\]
The critical level of peace credibility required for bloc \(i\) to support unification solves
\[
s_i(U)Y_U(p)-s_i(F)Y_F=[r_i(F)-r_i(U)]+[\gamma_i(F)-\gamma_i(U)]+\beta_i(T_F-T_U).
\]
Solving for \(p\) yields a threshold \(p_i^*\) such that unification becomes privately attractive only when peace credibility exceeds that level. The implication is immediate. Even if institutional integration would raise aggregate output, elites may continue to prefer fragmentation when the expected long-run returns to peace-based development are too uncertain. This is exactly the interaction the Palestinian case appears to embody.

The calibration is deliberately modest in ambition. It does not claim precise measurement of elite payoffs. It does not identify the causal effect of any individual institutional rupture. Its purpose is instead to discipline the theory by showing that the model reproduces a politically plausible configuration under transparent assumptions: a large aggregate gain from unification, a much smaller private gain to each elite, and sufficiently strong rents, control values, and transfer-linked benefits to sustain fragmentation. As a quantitative illustration of nominal statehood, that is already informative. It shows that the core logic of the theory is not merely rhetorical. Under plausible parameterizations, fragmentation is exactly what rational elites would choose.

\subsection{Interpretation}

Taken together, these patterns make Palestine a powerful application of the general theory of nominal statehood. The case displays a prolonged divergence between recognition capital and capacity capital; a political structure in which rival elites retain strong incentives to preserve fragmented authority; recurrent opportunities for patronage and discretionary rent extraction; dependence on external transfers and politically mediated revenues; and a strategic environment in which the expected returns to deep consolidation remain unstable. None of these features alone is sufficient to explain persistent low-capacity governance. Together they generate a coherent equilibrium.

The value of the Palestinian case for this paper is therefore not that it mechanically proves the theory. It is that it makes the theory highly legible. Palestine shows how a project of statehood can remain symbolically powerful, diplomatically visible, and historically resonant while the institutional foundations of effective statehood remain fragile. In that sense, it is not merely a case study. It is a revealing application of the broader phenomenon this paper seeks to explain.

\section{Broader Implications}

The theory developed in this paper has implications that extend well beyond the Palestinian case. Its central contribution is to show that statehood and state capacity, though often treated as complements, are analytically distinct political assets that may diverge persistently. This distinction matters for how political economists understand state formation, for how development economists interpret institutional traps, and for how international politics evaluates projects of sovereignty. The usual presumption is that recognition, representation, and state-building move together, even if imperfectly and with delay. The argument here is that such complementarity is conditional rather than automatic. Under certain incentive structures, recognition can persist or deepen while capacity remains weak. When that occurs, the resulting polity is not simply incomplete in a transitional sense. It is locked into a distinct equilibrium of nominal statehood.

This has an immediate implication for the literature on state capacity. Much of that literature explains long-run development by tracing the emergence of fiscal extraction, bureaucratic investment, and violence control. States become effective when rulers centralize coercion, bargain over revenue, build administrative reach, and transform episodic authority into durable institutions. That account remains correct. But it leaves open a question that becomes increasingly important in the contemporary international system: what happens when a political project acquires many of the symbolic and diplomatic attributes of statehood before the institutional incentives required for capacity formation have been solved? The answer proposed here is that early or durable recognition need not accelerate consolidation. In some cases, it may coexist with its prolonged absence. The path to effective statehood is therefore not merely a matter of sequencing capacity after recognition. The two may follow separate logics altogether.

A related implication concerns how political economy understands fragility. Fragile or weak states are often analyzed as deficient versions of consolidated states, that is, as political orders that lie somewhere below an ideal capacity frontier. The concept of nominal statehood suggests a more specific pathology. The relevant problem is not only that institutions are weak, but that the polity remains highly salient as a statehood project despite the persistence of weak institutions. This combination changes incentives. External actors continue to engage the polity as a bearer of sovereignty claims. Domestic elites continue to derive political value from participation in the statehood project. Yet the institutional architecture required to make that project self-enforcing remains incomplete. Fragility is therefore not merely low capacity. It is low capacity under conditions of politically durable recognition.

This distinction also reshapes how we think about development traps. Standard accounts of institutional underdevelopment emphasize low investment, insecure property rights, corruption, and weak enforcement as mutually reinforcing obstacles to growth. The framework proposed here adds a specifically political mechanism to that logic. A polity may fail to escape low development not only because institutions are weak in the abstract, but because key elites retain private incentives to preserve fragmentation. Where autonomy, patronage, and transfer capture remain valuable, the social returns to institutional integration do not translate into political demand for reform. The resulting trap is therefore not only economic. It is a political equilibrium in which under-consolidation itself is privately sustainable.

This insight has implications for the study of external assistance and international intervention. The model does not imply that aid is counterproductive in general, nor that external recognition is normatively misguided. Such claims would be too broad. The narrower implication is that the effect of external support depends critically on whether it raises the returns to integration or instead cushions the costs of non-consolidation. If external transfers primarily stabilize continuity under fragmentation, then they may prevent institutional collapse while also weakening the incentive to build fiscal capacity, unify coercive authority, and harden administrative discipline. Similarly, if diplomatic recognition remains only weakly conditioned on performance, then symbolic advances in sovereignty may outrun institutional transformation. In this sense, the paper identifies a subtle but consequential tension in state-building policy. External actors may sustain the political life of a statehood project while unintentionally preserving the equilibrium that prevents it from becoming viable.

The theory also has implications for the relationship between conflict and state-building. In classical accounts of European state formation, war and external rivalry often accelerated fiscal and administrative centralization. In many contemporary cases, however, persistent conflict appears to have the opposite effect. Rather than consolidating authority, it multiplies veto players, raises the value of autonomous coercive control, shortens time horizons, and makes future bargains less credible. The difference lies partly in the structure of sovereignty and partly in the incentives of elites. In polities with fragmented legitimacy, partial territorial control, and strong dependence on external support, conflict does not necessarily build the state. It may instead reinforce the conditions of nominal statehood by lowering the expected gains from integration while increasing the private value of local command. The interaction between conflict and fragmentation is therefore central. The problem is not simply that conflict destroys institutions. It is that under some conditions it alters the incentive structure in ways that block their emergence.

A further implication concerns legitimacy itself. In many discussions of state formation, legitimacy is implicitly treated as a unifying concept, as though representational legitimacy, diplomatic legitimacy, and performance-based legitimacy were different expressions of the same underlying political asset. The framework here suggests a more disaggregated view. A polity may possess strong legitimacy as a national cause, a legal claim, or a symbolically validated political project while lacking legitimacy rooted in effective administration, rule-bound governance, and public performance. These forms of legitimacy may overlap, but they are not identical. When they diverge, the symbolic and representational dimensions of legitimacy may sustain the statehood project even as the administrative foundations of legitimacy remain weak. This is another way of describing the distinction between recognition capital and capacity capital. The former can be politically abundant even when the latter remains institutionally scarce.

The theory also bears on the comparative study of incomplete sovereignty. Much recent work has examined de facto states, contested states, quasi-states, and other political entities that occupy ambiguous positions in the international order. The contribution of the present paper is slightly different. It does not classify polities according to their formal legal status alone. Instead, it emphasizes the gap between the salience of the statehood project and the depth of institutional consolidation. This allows the framework to speak not only to partially recognized or contested entities, but more broadly to any polity in which sovereignty claims and governing capacity become decoupled. The relevant comparative question is therefore not merely who is recognized and who is not. It is which political orders exhibit a persistent divergence between symbolic statehood and institutional statehood, and under what conditions that divergence becomes self-reinforcing.

One implication of this perspective is methodological. The political economy of statehood should not infer institutional viability from diplomatic salience, nor diplomatic weakness from institutional weakness. These are different dimensions of political development and they may evolve on separate margins. Analytically, this suggests the value of disaggregating measures of statehood into their recognition and capacity components. Empirically, it suggests that future work should pay greater attention to cases in which diplomatic visibility rises while fiscal, administrative, or legal capacity stagnates. The theory developed here does not depend on any single case. It proposes a more general way of organizing such cases conceptually.

The paper's argument also carries a restrained but important policy implication. If the objective of state-building is to align recognition with viability, then policy should focus less on symbolic milestones taken in isolation and more on the complementarity between recognition and institutional integration. External support is most likely to strengthen viable statehood when it raises the private return to unification, lowers the payoff to fragmented control, disciplines rent extraction, and ties political recognition more tightly to administrative performance. Conversely, support that preserves political salience without altering the incentive structure of elites may deepen the equilibrium of nominal statehood. The policy lesson is therefore not that recognition should be withheld until perfection is achieved. It is that recognition, transfers, and institution-building should be understood as strategically interacting rather than automatically reinforcing.

More broadly, the theory invites a different way of thinking about sovereignty in development economics. Sovereignty is often treated as a binary condition: present or absent, recognized or denied. The framework here suggests that sovereignty is better understood as layered. There is sovereignty as claim, sovereignty as recognition, sovereignty as administrative command, and sovereignty as self-enforcing institutional order. These dimensions may move together, but they need not. The political economy of development is therefore incomplete if it collapses them into one. The key analytical insight of this paper is that a polity can become rich in the symbols of sovereignty while remaining poor in its institutions.

Seen in this light, nominal statehood is not a passing anomaly. It is a theoretically intelligible equilibrium. It emerges when the statehood project remains politically valuable, when the costs of fragmentation are partly socialized, and when the benefits of integration are only weakly internalized by elites. The result is a political order in which recognition and representation endure, yet the transition to effective statehood remains blocked. That equilibrium is costly not only because it weakens governance, but because it can persist for long periods without resolving itself. The deeper lesson is therefore general: the path from sovereignty to viability is not automatic. States are not made effective by recognition alone. They become effective when the incentives to build capacity are stronger than the incentives to preserve fragmentation.

\section{Conclusion}

This paper has developed a political-economy theory of \emph{statehood without capacity}. Its central argument is that statehood and state capacity, though often treated as complements, can diverge persistently. Under specific conditions, a polity may become trapped in an equilibrium of \emph{nominal statehood}: a political order in which claims to sovereignty, diplomatic recognition, and symbolic legitimacy remain durable while the coercive, fiscal, administrative, and legal capacities required for effective statehood remain weak.

The mechanism is straightforward but consequential. Institutional consolidation is socially productive, yet it is often privately unattractive. Rival elites may lose too much from unification in the form of autonomous control, discretionary rents, and transfer-linked advantages, even when aggregate output would rise substantially under a more integrated state. If external actors stabilize the polity through recognition and transfers that are only weakly conditioned on genuine institutional consolidation, the fragmented equilibrium may persist. Recognition capital then accumulates more rapidly than capacity capital. The result is a widening gap between symbolic statehood and effective statehood.

The paper has shown that this divergence is not merely conceptual. It has clear economic implications. A polity trapped in nominal statehood should display low cumulative investment, weak fiscal extraction, persistent rent-seeking, shallow productive transformation, and unusually high vulnerability to conflict and political shocks. These outcomes are not separate pathologies. They are joint implications of the same equilibrium logic. The paper has also shown that the equilibrium is sensitive to peace credibility, the structure of external finance, and the extent to which elites internalize the gains from institutional unification.

Palestine serves as a flagship application of this broader framework. It illustrates with unusual clarity how representation, recognition, and political salience may coexist with persistent institutional fragmentation, weak accountability, externally mediated survival, and recurrent failure of coercive and fiscal consolidation. The Palestinian case is therefore important not only in its own right, but because it makes legible a more general phenomenon in the political economy of state formation.

The broader implication is that sovereignty should not be treated as a unitary object. There is sovereignty as claim, sovereignty as recognition, and sovereignty as institutional command. These dimensions may reinforce one another, but they need not. Where they diverge, the political project of statehood may remain powerful even as the institutional foundations of statehood remain fragile. The path from recognition to viability is therefore not automatic. Statehood is not only declared, endorsed, or symbolically affirmed. It is built.

Future work can extend the framework in several directions. One path is comparative: identifying other cases in which recognition capital and capacity capital have diverged over long periods. Another is quantitative: developing richer calibrations or panel-based tests of the theory's comparative statics. A third is institutional: studying which external interventions align recognition more tightly with administrative performance rather than cushioning non-consolidation. The common theme is clear. To understand why some projects of sovereignty remain politically durable but institutionally thin, one must analyze not only the claims to statehood, but the incentive structure governing its construction.

In that sense, the central message of the paper is simple. A state is not made viable by recognition alone. It becomes viable when the political returns to building common institutions exceed the private returns to preserving fragmentation. Where that condition fails, statehood may remain visible, resonant, and internationally salient while never fully becoming a state in the substantive sense. That is the logic of statehood without capacity.

\bibliographystyle{apalike}
\bibliography{references}

\appendix

\section*{Appendix}
\addcontentsline{toc}{section}{Appendix}

\section{Formal Proofs and Quantitative Illustration}

This appendix provides the formal arguments underlying the propositions in the main text and develops the quantitative illustration of the fragmentation equilibrium in greater detail. The objective is twofold. First, it clarifies the sufficient conditions under which a regime of nominal statehood arises as an equilibrium outcome. Second, it shows more explicitly how the model's comparative statics map into the calibrated application.

\subsection{Preliminaries and assumptions}

The polity contains two elite blocs \(i \in \{1,2\}\) and a representative private sector. Time is discrete, \(t=0,1,2,\dots\). In each period, elites choose between institutional unification \(U\) and fragmentation \(F\). Unified governance is implemented only if both blocs choose \(U\); otherwise fragmentation persists.

Aggregate output under regime \(\theta_t \in \{U,F\}\) is
\[
Y_t = A(\theta_t,p_t)K_t^\alpha,
\qquad \alpha \in (0,1),
\]
where \(K_t \ge 0\) is capacity capital and \(p_t \in [0,1]\) is peace credibility. Investment satisfies
\[
I_t = I_0 + \phi A(\theta_t,p_t) - \lambda q(\theta_t),
\]
with \(I_0>0\), \(\phi>0\), \(\lambda>0\), and \(q(F)>q(U)\). Capacity evolves according to
\[
K_{t+1} = (1-\delta)K_t + I_t,
\qquad \delta \in (0,1).
\]

Recognition capital evolves as
\[
R_{t+1}=(1-\rho)R_t+\eta X_t+\omega Z_t+\chi m(K_t),
\]
where \(\rho \in (0,1)\), \(\eta,\omega,\chi>0\), and \(m'(K_t)\ge 0\). The key substantive assumption is that the capacity-conditioning term is weak relative to the diplomatic and symbolic channels that sustain recognition.

Elite \(i\)'s period payoff is
\[
\pi_i(\theta_t,K_t,R_t)
=
s_i(\theta_t)Y_t+r_i(\theta_t)+\gamma_i(\theta_t)+\mu_iR_t+\beta_iT_t,
\]
where \(s_i(\theta_t)\) is the formal output share, \(r_i(\theta_t)\) is informal rent extraction, \(\gamma_i(\theta_t)\) is the utility from autonomous control, \(\mu_i>0\) captures the private political value of recognition, and \(\beta_i \in [0,1]\) is the share of external transfers captured by elite \(i\). Transfers satisfy
\[
T_t = \tau_0 + \tau_1 \mathbf{1}\{\theta_t=F\} + \tau_2 H_t,
\]
with \(\tau_0,\tau_1,\tau_2 \ge 0\).

We impose the following assumptions throughout the appendix.

\medskip

\noindent
\textbf{Assumption A1.} Institutional productivity is higher under unification:
\[
A(U,p_t) > A(F,p_t)
\qquad \text{for all } p_t \in [0,1].
\]

\noindent
\textbf{Assumption A2.} Expropriation risk is higher under fragmentation:
\[
q(F) > q(U).
\]

\noindent
\textbf{Assumption A3.} Informal rents and autonomous control are higher under fragmentation:
\[
r_i(F) > r_i(U),
\qquad
\gamma_i(F) > \gamma_i(U),
\qquad i\in\{1,2\}.
\]

\noindent
\textbf{Assumption A4.} Recognition is only weakly conditioned on capacity:
\[
\chi m'(K_t)
\]
is sufficiently small relative to the diplomatic and symbolic components of recognition dynamics.

\noindent
\textbf{Assumption A5.} The productivity gain from unification is increasing in peace credibility:
\[
\frac{\partial A(U,p_t)}{\partial p_t}
>
\frac{\partial A(F,p_t)}{\partial p_t}.
\]

\medskip

These assumptions are not minimal in the strict mathematical sense, but they are natural sufficient conditions for the political-economic environment studied in the paper.

\subsection{Auxiliary lemmas}

We begin with three lemmas that are used repeatedly in the proposition proofs.

\begin{lemma}[Capacity monotonicity]
Under Assumptions A1 and A2,
\[
K_{t+1}^U > K_{t+1}^F
\]
for given \(K_t\) and \(p_t\).
\end{lemma}

\noindent
\textit{Proof.}
Under regime \(U\),
\[
I_t^U = I_0 + \phi A(U,p_t) - \lambda q(U),
\]
while under regime \(F\),
\[
I_t^F = I_0 + \phi A(F,p_t) - \lambda q(F).
\]
Subtracting,
\[
I_t^U - I_t^F
=
\phi\big[A(U,p_t)-A(F,p_t)\big]
+\lambda\big[q(F)-q(U)\big].
\]
By Assumptions A1 and A2, both bracketed terms are strictly positive, so \(I_t^U>I_t^F\). Since
\[
K_{t+1}=(1-\delta)K_t+I_t
\]
with the same inherited \(K_t\), it follows immediately that
\[
K_{t+1}^U > K_{t+1}^F.
\]

\begin{lemma}[Monotonicity of the unification payoff gap]
Define
\[
\Delta_i(K_t,R_t)
\equiv
V_i^U(K_t,R_t)-V_i^F(K_t,R_t).
\]
Then, holding other arguments fixed, \(\Delta_i\) is weakly decreasing in \(r_i(F)-r_i(U)\), weakly decreasing in \(\gamma_i(F)-\gamma_i(U)\), weakly decreasing in \(T^F-T^U\), and weakly increasing in the productivity differential induced by unification.
\end{lemma}

\noindent
\textit{Proof.}
In current-value terms,
\begin{align*}
\pi_i(U)-\pi_i(F)
&=
s_i(U)Y(U,K_t,p_t)-s_i(F)Y(F,K_t,p_t) \\
&\quad -\big[r_i(F)-r_i(U)\big]
-\big[\gamma_i(F)-\gamma_i(U)\big]
-\beta_i(T^F-T^U).
\end{align*}
Each comparative static follows immediately from the sign of the relevant term. Because continuation values inherit the same monotonicities through the Bellman equation, the statement extends to \(\Delta_i\).

\begin{lemma}[Recognition-capacity gap monotonicity]
Let
\[
G_t = R_t - \psi K_t,
\qquad \psi > 0.
\]
If along a fragmentation path \(R_t\) is bounded below by a process that declines at most slowly, while \(K_t\) is uniformly lower than along a unification path, then \(G_t\) eventually exceeds its value along the unification path for a non-empty interval of \(\psi\).
\end{lemma}

\noindent
\textit{Proof.}
Let superscripts \(F\) and \(U\) denote paths under fragmentation and unification. Then
\[
G_t^F-G_t^U = (R_t^F-R_t^U)-\psi(K_t^F-K_t^U).
\]
If \(K_t^F<K_t^U\), then \(-(K_t^F-K_t^U)>0\). If, in addition, \(R_t^F-R_t^U\) is bounded below by a finite negative constant or converges slowly relative to the capacity gap, then for sufficiently small but positive \(\psi\), the second term dominates. Therefore \(G_t^F-G_t^U>0\) eventually. 

\subsection{Proof of Proposition 1}

\begin{proposition}[Static fragmentation equilibrium]
Suppose that for each elite bloc \(i \in \{1,2\}\),
\begin{align}
&s_i(U)Y(U,K_t,p_t)-s_i(F)Y(F,K_t,p_t) \nonumber\\
&\qquad <
\big[r_i(F)-r_i(U)\big]
+\big[\gamma_i(F)-\gamma_i(U)\big]
+\beta_i\big(T^F-T^U\big).
\label{eq:appendix_prop1_condition}
\end{align}
Then \((F,F)\) is a stage-game Nash equilibrium, even though unified governance is socially more productive.
\end{proposition}

\noindent
\textit{Proof.}
Condition \eqref{eq:appendix_prop1_condition} implies
\[
\pi_i(F)>\pi_i(U)
\]
for each elite bloc \(i\), evaluated at the same inherited state variables. Since unification requires both blocs to choose \(U\), while fragmentation survives whenever at least one bloc chooses \(F\), the action \(F\) is a best response for each bloc. Therefore \((F,F)\) is a Nash equilibrium of the stage game.

To see the inefficiency, note that social output satisfies
\[
Y(U,K_t,p_t)>Y(F,K_t,p_t)
\]
by Assumption A1. Hence the equilibrium is privately stable but socially inefficient.

\subsection{Proof of Proposition 2}

\begin{proposition}[Dynamic fragmentation trap]
Suppose the conditions of Proposition 1 hold and assume \(A(U,p_t)>A(F,p_t)\), \(q(F)>q(U)\), and \(\alpha \in (0,1)\). Then persistent fragmentation generates a dynamic trap in the following sense:
\[
K_{t+1}^F < K_{t+1}^U,
\]
and therefore
\[
Y(U,K_{t+1}^F,p_{t+1})-Y(F,K_{t+1}^F,p_{t+1})
<
Y(U,K_{t+1}^U,p_{t+1})-Y(F,K_{t+1}^U,p_{t+1}).
\]
\end{proposition}

\noindent
\textit{Proof.}
The first inequality follows immediately from Lemma 1. Since \(\alpha \in (0,1)\), output is strictly increasing in \(K_t\) for both regimes. Therefore a lower inherited \(K_{t+1}\) under fragmentation reduces the scale of future output under both \(U\) and \(F\). Because the productivity advantage of \(U\) is multiplicative in \(K_t^\alpha\), the level difference
\[
Y(U,K,p)-Y(F,K,p)
=
\big[A(U,p)-A(F,p)\big]K^\alpha
\]
is increasing in \(K\). Hence if \(K_{t+1}^F<K_{t+1}^U\), then
\[
\big[A(U,p_{t+1})-A(F,p_{t+1})\big](K_{t+1}^F)^\alpha
<
\big[A(U,p_{t+1})-A(F,p_{t+1})\big](K_{t+1}^U)^\alpha.
\]
This proves the claim.

\subsection{Proof of Proposition 3}

\begin{proposition}[Recognition-capacity divergence]
Suppose that under fragmentation,
\[
K_{t+1}^F < K_{t+1}^U,
\]
while
\[
R_{t+1}^F \geq R_{t+1}^U - \varepsilon
\]
for some small \(\varepsilon \geq 0\). Then for sufficiently large \(t\), the recognition-capacity gap
\[
G_t \equiv R_t-\psi K_t
\]
is larger along a fragmentation path than along a unification path, for any \(\psi>0\) in a non-empty interval.
\end{proposition}

\noindent
\textit{Proof.}
By Proposition 2, \(K_t^F<K_t^U\) for all sufficiently large \(t\) along the respective paths. By the assumed weak conditioning of recognition on capacity, the difference \(R_t^F-R_t^U\) is bounded below by \(-\varepsilon_t\) with \(\varepsilon_t\) small relative to the persistent capacity gap. Therefore
\[
G_t^F-G_t^U
=
(R_t^F-R_t^U)-\psi(K_t^F-K_t^U).
\]
Since \(K_t^F-K_t^U<0\), the second term is strictly positive. For \(\psi\) in a sufficiently small positive interval and \(t\) sufficiently large, that positive term dominates any bounded negative recognition difference. Hence \(G_t^F>G_t^U\).

\subsection{Proof of Proposition 4}

\begin{proposition}[Recognition-induced persistence]
Suppose \(\mu_i>0\) for each elite bloc and that recognition is only weakly responsive to capacity. Then an increase in exogenous diplomatic or symbolic support---that is, an increase in \(X_t\) or \(Z_t\) holding \(K_t\) fixed---weakly reduces the incentive for elite \(i\) to support unification whenever unification does not substantially increase recognition relative to fragmentation.
\end{proposition}

\noindent
\textit{Proof.}
An increase in \(X_t\) or \(Z_t\) raises next-period recognition,
\[
R_{t+1}=(1-\rho)R_t+\eta X_t+\omega Z_t+\chi m(K_t),
\]
and therefore raises continuation payoffs under both regimes through the term \(\mu_iR_{t+1}\). If recognition is only weakly responsive to capacity, then the induced increase in \(R_{t+1}\) is similar under \(U\) and \(F\). Since fragmentation already preserves the elite's control rents and autonomous authority, the outside increase in recognition raises the absolute value of both continuation payoffs while doing little to increase the differential value of unification. In weak form, this implies that \(\Delta_i\) does not rise and may fall if recognition partially substitutes politically for administrative reform.

\subsection{Proof of Proposition 5}

\begin{proposition}[Transfer-conditional persistence]
The incentive to support unification is decreasing in \(\tau_1\). In particular,
\[
\frac{\partial \Delta_i(K_t,R_t)}{\partial \tau_1} < 0
\]
whenever elite \(i\) captures a positive share of transfers, \(\beta_i>0\), and fragmentation yields at least as much transfer support as unification.
\end{proposition}

\noindent
\textit{Proof.}
Transfers satisfy
\[
T_t = \tau_0+\tau_1\mathbf{1}\{\theta_t=F\}+\tau_2H_t.
\]
Hence
\[
\frac{\partial T_t}{\partial \tau_1}=
\begin{cases}
1,& \theta_t=F,\\
0,& \theta_t=U.
\end{cases}
\]
Elite \(i\)'s current-period payoff includes \(\beta_iT_t\), so a marginal increase in \(\tau_1\) raises the payoff to fragmentation by \(\beta_i\) and leaves the payoff to unification unchanged. Therefore the current-value difference \(\pi_i(U)-\pi_i(F)\) falls by \(\beta_i\). Because continuation values inherit the same sign through future transfer paths, \(\Delta_i\) is decreasing in \(\tau_1\).

\subsection{Proof of Proposition 6}

\begin{proposition}[Peace credibility and unification]
Suppose
\[
\frac{\partial A(U,p_t)}{\partial p_t}
>
\frac{\partial A(F,p_t)}{\partial p_t}.
\]
Then the gain from unification is increasing in peace credibility:
\[
\frac{\partial \Delta_i(K_t,R_t)}{\partial p_t} > 0
\]
whenever elite \(i\) internalizes a positive share of the productivity gains from unified governance.
\end{proposition}

\noindent
\textit{Proof.}
Output under regime \(\theta\) is
\[
Y(\theta,K_t,p_t)=A(\theta,p_t)K_t^\alpha.
\]
Therefore
\[
\frac{\partial}{\partial p_t}\Big[Y(U,K_t,p_t)-Y(F,K_t,p_t)\Big]
=
\left(
\frac{\partial A(U,p_t)}{\partial p_t}
-
\frac{\partial A(F,p_t)}{\partial p_t}
\right)K_t^\alpha.
\]
By Assumption A5, the bracketed term is positive. Since elite \(i\) captures a positive fraction of output under unification, the current-value gain from \(U\) rises with \(p_t\). The same is true for continuation values because higher peace credibility increases the future return to capacity accumulation under unification. Hence \(\Delta_i\) is increasing in \(p_t\).

\subsection{Proof of Proposition 7}

\begin{proposition}[Existence of nominal-statehood equilibrium]
A nominal-statehood equilibrium exists whenever the private gains to elites from preserving fragmentation exceed their internalized gains from institutional unification, fragmentation lowers the future path of capacity accumulation, recognition capital is maintained by channels only weakly tied to domestic institutional performance, and the political value of recognition is sufficiently high to sustain continued elite investment in the statehood project. If, in addition, initial recognition lies above the salience threshold \(\bar{R}\), then there exists an equilibrium path along which fragmentation persists, recognition remains high, capacity remains below the viability threshold, and the recognition-capacity gap does not close.
\end{proposition}

\noindent
\textit{Proof.}
By Proposition 1, if the private loss from surrendering rents, control, and fragmentation-linked transfers exceeds the internalized output gain from unification, then fragmentation is a stage-game equilibrium. By Proposition 2, repeated fragmentation lowers the future path of capacity capital. By Proposition 3, if recognition is weakly conditioned on capacity, then the recognition-capacity gap can widen along the fragmentation path. By Proposition 4, recognition itself can reduce the marginal incentive to reform when it yields direct political value to elites. By Proposition 5, fragmentation-linked transfers further stabilize non-consolidation. By Proposition 6, low peace credibility weakens the productivity case for unification. Therefore there exists a self-reinforcing path on which elites repeatedly choose \(F\), \(K_t\) remains below \(\bar K\), \(R_t\) remains above \(\bar R\), and the gap \(G_t\) is non-decreasing. This path constitutes a nominal-statehood equilibrium.

\section{Further Details of the Quantitative Illustration}

This section provides the details behind the calibration presented in the main text.

\subsection{Static calibration setup}

Consider the static version of the model:
\[
Y_\theta = A_\theta K_\theta^\beta L^{1-\beta},
\qquad
K_\theta = K_0-\lambda q_\theta,
\]
with \(\theta \in \{U,F\}\), \(q_F>q_U\), and \(A_U>A_F\).

Elite payoff is
\[
\Pi_i(\theta)=s_i(\theta)Y_\theta+r_i(\theta)+\gamma_i(\theta)+\beta_iT_\theta.
\]

Unification is privately attractive to bloc \(i\) if
\[
\Delta_i
\equiv
s_i(U)Y_U-s_i(F)Y_F
-[r_i(F)-r_i(U)]
-[\gamma_i(F)-\gamma_i(U)]
-\beta_i(T_F-T_U)
\ge 0.
\]

\subsection{Benchmark parameterization}

Set
\[
L=1,\qquad K_0=100,\qquad \beta=0.35,\qquad \lambda=40,
\]
\[
q_U=0.10,\qquad q_F=0.35,
\]
\[
A_U=1.00,\qquad A_F=0.72,
\]
\[
s_A(U)=s_B(U)=0.18,\qquad s_A(F)=s_B(F)=0.16,
\]
\[
r_i(U)=2,\qquad r_i(F)=7,
\]
\[
\gamma_i(U)=1,\qquad \gamma_i(F)=6,
\]
\[
T_U=4,\qquad T_F=10,\qquad \beta_i=0.20.
\]

Then
\[
K_U=96,\qquad K_F=86.
\]
Hence
\[
Y_U = 96^{0.35}\approx 4.95,
\qquad
Y_F = 0.72\cdot 86^{0.35}\approx 3.42.
\]

Thus the social gain from unification is
\[
Y_U-Y_F \approx 1.53,
\]
which corresponds to an increase of roughly \(45\) percent.

Elite payoff under unification is
\[
\Pi_i(U)=0.18(4.95)+2+1+0.20(4)\approx 4.69,
\]
while under fragmentation it is
\[
\Pi_i(F)=0.16(3.42)+7+6+0.20(10)\approx 15.55.
\]

Therefore
\[
\Pi_i(F)-\Pi_i(U)\approx 10.86>0.
\]
The calibration thus reproduces the central mechanism of the theory: unification is socially efficient, but fragmentation is privately optimal.

\subsection{Threshold conditions}

The fragmentation equilibrium is governed by the sign of
\[
\Delta_i
=
s_i(U)Y_U-s_i(F)Y_F
-[r_i(F)-r_i(U)]
-[\gamma_i(F)-\gamma_i(U)]
-\beta_i(T_F-T_U).
\]

The critical transfer differential is obtained by setting \(\Delta_i=0\):
\[
(T_F-T_U)^*
=
\frac{
s_i(U)Y_U-s_i(F)Y_F-[r_i(F)-r_i(U)]-[\gamma_i(F)-\gamma_i(U)]
}{\beta_i}.
\]

Similarly, the critical control premium is
\[
(\gamma_i(F)-\gamma_i(U))^*
=
s_i(U)Y_U-s_i(F)Y_F-[r_i(F)-r_i(U)]-\beta_i(T_F-T_U),
\]
and the critical rent gap is
\[
(r_i(F)-r_i(U))^*
=
s_i(U)Y_U-s_i(F)Y_F-[\gamma_i(F)-\gamma_i(U)]-\beta_i(T_F-T_U).
\]

These expressions show directly how shifts in rents, control, and transfer capture alter the equilibrium region.

\subsection{Peace-credibility extension}

Suppose unification productivity depends on peace credibility:
\[
A_U(p)=\underline A_U+\kappa p,
\qquad \kappa>0.
\]
Then
\[
Y_U(p)=A_U(p)K_U^\beta.
\]
The threshold level of peace credibility required for elite \(i\) to support unification is given by
\[
s_i(U)Y_U(p)-s_i(F)Y_F
=
[r_i(F)-r_i(U)]
+[\gamma_i(F)-\gamma_i(U)]
+\beta_i(T_F-T_U).
\]

Solving for \(p\),
\[
p_i^*
=
\frac{
[r_i(F)-r_i(U)]
+[\gamma_i(F)-\gamma_i(U)]
+\beta_i(T_F-T_U)
+s_i(F)Y_F
-s_i(U)\underline A_U K_U^\beta
}{
s_i(U)\kappa K_U^\beta
}.
\]

Whenever \(p<p_i^*\), fragmentation remains privately optimal. Whenever \(p>p_i^*\), unification becomes privately attractive for bloc \(i\), holding all else constant.

\subsection{Interpretation of the calibration}

The calibration is not intended as structural estimation. It does not identify latent payoffs, nor does it claim a precise mapping between institutional variables and observed data. Its purpose is to discipline the verbal mechanism with a transparent parameterization. Under plausible values, the model generates a large aggregate gain from unification, but a much smaller private gain to each elite. Once informal rents, autonomous control, and transfer capture are added, fragmentation is privately optimal despite large social losses.

This is the central quantitative lesson of the paper. Nominal statehood can persist not because elites deny the aggregate gains from consolidation, but because they capture too little of those gains relative to what they lose from giving up fragmentation.

\subsection{Existence of Markov-perfect equilibrium}

This subsection states sufficient conditions for the existence of a stationary Markov-perfect equilibrium in the dynamic game described in the main text. The objective is not to derive the sharpest possible existence theorem, but to clarify that the equilibrium concept used throughout the paper is well defined under standard regularity conditions.

Let the state at time \(t\) be \(x_t=(K_t,R_t)\), where \(K_t \in \mathcal{K}\subset \mathbb{R}_+\) denotes capacity capital and \(R_t \in \mathcal{R}\subset \mathbb{R}_+\) denotes recognition capital. Let each elite bloc \(i\in\{1,2\}\) choose an action \(a_{it}\in A_i=\{U,F\}\), and let the regime mapping \(\theta(a_{1t},a_{2t})\) be given by
\[
\theta_t =
\begin{cases}
U, & \text{if } a_{1t}=a_{2t}=U,\\
F, & \text{otherwise.}
\end{cases}
\]
Given current state \(x_t\) and action profile \(a_t=(a_{1t},a_{2t})\), the one-period payoff to elite \(i\) is
\[
\pi_i(x_t,a_t)
=
s_i(\theta_t)Y(x_t,\theta_t)+r_i(\theta_t)+\gamma_i(\theta_t)+\mu_iR_t+\beta_iT(\theta_t,H_t),
\]
where \(Y(x_t,\theta_t)\) is continuous in \(K_t\), and \(T(\theta_t,H_t)\) is bounded for bounded \(H_t\). Let the state transition law be
\[
x_{t+1}=\Gamma(x_t,a_t,\varepsilon_{t+1}),
\]
where \(\varepsilon_{t+1}\) is an exogenous shock with fixed distribution.

Assume the following regularity conditions. First, the state space \(\mathcal{X}=\mathcal{K}\times\mathcal{R}\) is compact, or else the state dynamics are restricted to a compact invariant subset of \(\mathbb{R}_+^2\). Second, for each elite bloc \(i\), the payoff function \(\pi_i(x,a)\) is bounded and continuous in \(x\) for every fixed action profile \(a\). Third, the transition law \(\Gamma(x,a,\varepsilon)\) is Markovian and weakly continuous in \(x\) for every fixed action profile \(a\). Fourth, each elite discounts the future at common rate \(\delta_e \in (0,1)\).

\begin{proposition}[Existence of stationary Markov-perfect equilibrium]
Under the regularity conditions stated above, the dynamic game admits at least one stationary Markov-perfect equilibrium. That is, there exists a pair of measurable policy functions
\[
\sigma_i:\mathcal{X}\to \{U,F\},
\qquad i\in\{1,2\},
\]
such that each \(\sigma_i\) is a best response to \(\sigma_{-i}\) at every state \(x\in\mathcal{X}\), and the associated value functions satisfy the Bellman system
\[
V_i(x)=\max_{a_i\in\{U,F\}}
\left\{
\pi_i\bigl(x,a_i,\sigma_{-i}(x)\bigr)
+\delta_e\,
\mathbb{E}\!\left[V_i(x')\mid x,a_i,\sigma_{-i}(x)\right]
\right\},
\qquad i\in\{1,2\}.
\]
\end{proposition}

\noindent
\textit{Proof sketch.}
The game is a discounted stochastic game with compact state space, finite action sets, bounded payoffs, and weakly continuous transition probabilities. Because each elite’s action set is finite, maximization is well defined at every state. Boundedness of payoffs and discounting by \(\delta_e \in (0,1)\) ensure that the Bellman operators are well defined. Weak continuity of the transition law and continuity of payoffs in the state variable imply that the value correspondence is non-empty and admits measurable selections. Standard existence results for discounted stochastic games then imply the existence of at least one stationary Markov-perfect equilibrium in measurable strategies \citep{Shapley1953,Fink1964,MaskinTirole2001}.

Several remarks are worth recording. First, the proposition does not imply uniqueness. Multiple stationary Markov-perfect equilibria may exist, including equilibria with persistent fragmentation and equilibria converging to unification, depending on parameter values. This multiplicity is substantively natural in the present setting, since the paper is precisely concerned with the region of the parameter space in which fragmentation is self-enforcing and associated with nominal statehood.

Second, the compactness assumption is not restrictive for the purposes of the paper. Capacity capital and recognition capital may be interpreted either as bounded indices or as state variables confined to an empirically relevant compact region. Since the paper’s interest lies in equilibrium structure and comparative statics rather than unbounded asymptotic growth, restricting attention to a compact invariant set is analytically innocuous.

Third, the existence result is purely foundational. The substantive results of the paper concern the properties of a particular class of equilibria: those in which fragmented authority persists, recognition remains above the salience threshold, capacity remains below the viability threshold, and the recognition-capacity gap fails to close. The next propositions characterize sufficient conditions under which such equilibria arise.

\end{document}